\documentclass[10pt, doublecolumn]{IEEEtran}
\pdfoutput=1
\usepackage{cite}
\usepackage{indentfirst}
\usepackage{graphicx}
\usepackage{changepage}
\usepackage{stfloats}
\usepackage{amsfonts,amssymb}
\usepackage{bm}
\usepackage{algorithm}
\usepackage{algorithmic}
\usepackage{amsmath}
\usepackage{amsmath,amssymb,amsfonts}
\usepackage{times}
\usepackage{url}
\usepackage{cite}	
\usepackage{textcomp}
\usepackage{color}
\usepackage{setspace}
\usepackage{enumitem}
\usepackage{float}
\usepackage{diagbox}
\usepackage[T1]{fontenc}
\usepackage[utf8]{inputenc}
\usepackage{authblk}
\usepackage{threeparttable}
\usepackage{booktabs}
\usepackage{footnote}
\usepackage{graphicx}
\usepackage{pifont}
\usepackage{subfigure}
\usepackage{mathrsfs}
\usepackage{bbm}
\usepackage{dsfont}
\usepackage{colortbl}
\usepackage{multirow}
 \usepackage{booktabs,threeparttable}
\usepackage{makecell}
\usepackage[dvipsnames,table]{xcolor}

\definecolor{HeaderBlue}{HTML}{2C3E50}   
\definecolor{RowGray}{HTML}{F2F4F4}     
\definecolor{RowWhite}{HTML}{FFFFFF}    

\setcellgapes{3pt}
\makegapedcells

\allowdisplaybreaks[4]

\newenvironment{proof}{{\indent  \indent \it Proof:}}{\hfill $\blacksquare$}

\begin{document}
\title{Constellation Selection and Power Control for OFDM-based ISAC: From Theory to Prototype}

\author{
	Kaitao Meng, \textit{Member, IEEE}, Kawon Han, \textit{Member, IEEE}, Christos Masouros, \textit{Fellow, IEEE}, and Fan Liu , \textit{Senior Member, IEEE}
	\thanks{Preliminary versions of this paper were presented at the IEEE International Conference on Communications (ICC) 2026 \cite{Kaitao2026ICC}}
	\thanks{The work of K. Meng was supported in part by UKRI under Grant EP/Y02785X/1. This work of K. Han was supported by the National Research Foundation of Korea (NRF) grant by the Korea government (MSIT) (RS-2026-25497418). The work of C. Masouros is supported by the Smart Networks and Services Joint Undertaking (SNS JU) project 6G-MUSICAL under Grant Agreement No. 101139176, and by the Engineering and Physical Sciences Research Council UK-India project ICON with Grant Agreement No UKRI859. This work of F. Liu was supported in part by the National Natural Science Foundation of China (NSFC) under Grant 62522107 and Grant 62331023.}
	\thanks{K. Meng is with the Department of Electrical and Electronic Engineering, University of Manchester, Manchester, UK (email: kaitao.meng@manchester.ac.uk). K. Han, and C. Masouros are with the Department of Electronic and Electrical Engineering, University College London, London, UK (emails: \{kawon.han, c.masouros\}@ucl.ac.uk). F. Liu is with the National Mobile Communications Research Laboratory, Southeast University, China (emails: fan.liu@seu.edu.cn). (Corresponding author: Kawon Han)   
	
}}


\maketitle


\begin{abstract}
Integrated sensing and communication (ISAC) techniques can leverage existing, wide-coverage communication networks to perform sensing tasks, enabling large-scale and low-cost target sensing. However, the inherent randomness of communication data payloads introduces undesired sidelobes in the ambiguity function that may degrade target detection and parameter estimation performance. This paper develops a communication-centric ISAC framework that is standards-compliant and compatible with existing devices. Specifically, we propose a low-complexity constellation selection scheme over a {\textit{finite}}, {\textit{off-the-shelf}} alphabet, achieving an efficient sensing-communication trade-off without custom waveforms or frame-structure changes. To this end, we analyze two classical sensing receivers including matched filtering (MF) and reciprocal filtering (RF) for ranging measurements, {{and derive closed-form sensing laws that link constellation statistics to sensing performance. Under any finite-alphabet constellation combination, MF sidelobes depend on the weighted sum of the kurtosis values of the per-subcarrier constellations, while RF noise enhancement depends on the inverse second moment of the transmit symbol, providing a tractable expression for tuning the sensing-communication trade-off.}} Guided by these insights, we propose a low-complexity design that retains a finite constellation set and shapes power under communication quality constraints. We prove that in flat-fading channels, there exists a Pareto-optimal solution that activates no more than three constellations. For frequency-selective channels, a bilevel algorithm with closed-form inner updates attains near-optimal performance while sharply reducing computational complexity. We validate the entire theoretical pipeline with numerical simulations as well as experimental results. The results corroborate the theory, constellation mixing and coherent multi-symbol processing reliably reveal markedly weaker targets at practical signal-to-noise ratios. 
\end{abstract}   


\newtheorem{thm}{\bf Lemma}
\newtheorem{remark}{\bf Remark}
\newtheorem{Pro}{\bf Proposition}
\newtheorem{theorem}{\bf Theorem}
\newtheorem{Assum}{\bf Assumption}
\newtheorem{Cor}{\bf Corollary}
\newtheorem{Def}{\bf Definition}
\newtheorem{assumption}{\bf Assumption}

\section{Introduction}
High-precision localization and high-rate data delivery underpin emerging applications such as autonomous driving and advanced augmented reality \cite{lu2021real, zhang2022artificial}. Operating separate sensing and communication (S\&C) stacks duplicates front-end hardware and frame structures, leads to uncoordinated spectrum use and mutual interference, motivating integrated sensing and communication (ISAC) techniques \cite{Zhang2021OverviewSignal, Liu2022SurveyFundamental, Meng2023SensingAssisted, 10473676, 10216343}. By reusing infrastructure, spectrum, and configurable waveforms, ISAC enables simultaneous data transmission and echo-based perception with quantifiable gains in spectral, energy, and cost efficiency \cite{Cui2021Integrating, meng2024integrated}. The unified architecture further permits joint design of waveform \cite{han2025signaling}, novel frame structure \cite{Hua20243DMultiTarget}, and resource allocation to balance rate, reliability, and sensing resolution under common constraints \cite{Wang2022PartiallyConnected, Lu2024Integrated}. Beyond link-level co-design, recent work elevates ISAC to the network-level cooperation for interference control \cite{Meng2024CooperativeISACMag}, cooperative sensing \cite{11412372, Meng2024CooperativeTWC}, and cross-link scheduling \cite{han2025network}. Consistent with this trajectory, ISAC has been identified by the ITU as a representative 6G usage scenario.

Given the ubiquity of cellular infrastructure already providing wide-area coverage, communication-centric ISAC designs that reuse existing waveforms and frame structures are attractive for near-term deployment because they adhere to existing standards and remain compatible with legacy devices \cite{meng2023network, Wei2023Integrated}. 
{{A practical obstacle, however, is that data payload symbols are inherently random \cite{Lu2024Random}. In an orthogonal frequency-division multiplexing (OFDM) sensing receiver, this randomness appears in the matched-filter/ambiguity-function output as {sidelobes}, i.e., undesired off-peak responses away from the true target delay. If these sidelobes become strong, they can raise the apparent clutter floor, mask weak targets, or even create ghost peaks \cite{Chitre2020AFShaping}. Existing efforts have shown that such randomness-induced sidelobes can indeed be reduced through precoder design \cite{Li2025MIMOOFDM} or sensing-oriented resource-block design \cite{Li2025SensingOriented}. However, these approaches typically rely on either instantaneous transmit-signal optimization or dedicated sensing resources. 
Importantly, since the sidelobes originate from the randomness of the communication payload, they can also be mitigated at the statistical level. In particular, constellation design provides a natural lever to control the payload statistics that govern sidelobe behavior, thereby improving the sensing-communication trade-off.}} Constellation design operates at the modulation layer, where optimizing the symbol alphabet shapes the statistics that set ambiguity-function sidelobes and effective noise. Although constellation optimization offers a lightweight lever, realizing it in an implementable, standards-aligned manner is non-trivial and remains an open problem \cite{3gpp_ts}. {{This motivates deeper investigation of standards-compliant constellation shaping, mixing, and selection strategies that preserve the waveform, alphabet/bit-mapping, and 5G new radio (NR) time-frequency resource grid.}}

Recent work has sharpened the link between data-modulation (constellation mapping) and sensing performance. {In particular under single antenna transmission and average-power normalization, \cite{Liu2025Opitimal} establishes the ranging optimality of OFDM, and multiple studies show that, under uniform average power allocation, PSK constellations generally yield lower range-sidelobe levels than QAM constellations, thereby improving delay resolution and ambiguity suppression \cite{zhang2023input, keskin2025fundamental}. 
Specifically, ambiguity-function variance is governed by constellation kurtosis, so high-kurtosis signaling leads to larger sidelobe fluctuations \cite{du2024reshaping}. Intuitively, kurtosis is a fourth-order statistic that quantifies how amplitude-dispersed a symbol distribution is. Constant-modulus constellations, such as PSK, typically have low kurtosis, whereas non-constant-modulus constellations, such as 64QAM, generally exhibit higher kurtosis due to their larger amplitude variations. As a result, high-kurtosis signaling tends to produce more random matched-filter sidelobes, which is undesirable for stable ranging.}
 Complementary to these results, \cite{Liu2025Uncovering} analyzes sensing performance under a uniform signaling setup, where all subcarriers use the same constellation with fixed power, and proves that OFDM is optimal for ranging. Existing results mainly rely on a homogeneous signaling assumption with the same constellation across all active subcarriers, leaving the sensing implications and sensing-communication trade-offs of subcarrier-wise constellation assignment largely unexplored.  {{In parallel, early studies pursue constellation optimization approaches that modify the signal alphabet itself to target theoretical performance limits \cite{Hu2025Learning, Geiger2025JCS, Du2023GCWkshpsPCS, Yang2024TWC_RandomWF}. However, such redesigns are difficult to deploy in practical systems because they depart from standardized constellations, may require signalling new constellation parameters and bit mappings, necessitate receiver demapper redesign, and can invalidate existing link-adaptation tables.}} These limitations motivate a standards-compliant alternative: Retain the existing constellation set, select suitable members per scenario, and jointly shape power to satisfy communication QoS while improving sensing. To our knowledge, a low-complexity method that performs this adaptive constellation selection and subcarrier power control has been largely underexplored. 

To address the above gap, in this paper we propose a low-complexity, standards-compliant scheme for adaptive constellation selection and subcarrier power control. {{Unlike conventional communication power control, which targets rate or QoS alone \cite{yates2002framework}, our design introduces an explicit degree of freedom to balance sensing and communication via joint power and constellation control. Specifically, we model the payload explicitly as a random process and develop an OFDM-based, communication-centric ISAC co-design that is compatible with existing standards. We analyze two canonical sensing receiver chains: matched filtering (MF) \cite{leuschen2002matched} and reciprocal filtering (RF) \cite{Wojaczek2019Reciprocal}. RF suppresses target/clutter sidelobes by equalizing data dependence but incurs a waveform-dependent SNR loss relative to MF, yielding a tunable trade-off. While \cite{11202391} contrasts MF and RF sidelobes as a function of target SNR, a rigorous, end-to-end characterization of how per-subcarrier constellation design and power allocation propagate through MF/RF chains to determine OFDM sensing performance (e.g., range-sidelobe levels, estimation accuracy) remains largely open.}}
To address this issue, we derive closed-form performance laws that expose the roles of constellation moments: (i) MF sidelobe power is dictated by the per-subcarrier fourth moment; (ii) RF noise amplification is governed by the inverse second moment of modulated signals. The analysis extends to multi-symbol coherent integration and yields the expected gain, which we exploit for super-resolution delay estimation via a Matrix Pencil estimator. 
Rather than redesigning alphabets, we select from a finite, standards-compliant constellation pool and shape power. This is compatible with existing OFDM protocols and demappers.
The main contributions of this paper are summarized as follows:
\begin{itemize}
	\item {{We propose a low-complexity, standards-compliant scheme for adaptive constellation selection and subcarrier power control in ISAC systems, by mixing a {small} set of off-the-shelf constellations across subcarriers and allocating power under communication quality constraints. We derive closed-form sensing laws that decouple constellation statistics from power allocation.}} For MF, sidelobes are governed by the subcarrier-dependent fourth-order moments; for RF, the effective noise is governed by the subcarrier-dependent inverse second moment of the transmitted symbols. {{To the best of our knowledge, such a receiver-oriented characterization of subcarrier-dependent waveform design has not been explicitly established in the prior OFDM-ISAC literature.}}
	
	\item {{Based on the derived sensing laws, we establish receiver-aware power-allocation principles for payload-bearing OFDM sensing. For MF, less power should be assigned to subcarriers employing constellations with larger kurtosis, while for RF, more power should be allocated to subcarriers with larger inverse second moments. Moreover, within any subset of subcarriers sharing the same constellation statistics, equal power allocation is proved to be optimal, revealing an explicit structural property of the joint sensing-communication design. }}

	\item {{We prove that, in flat-fading channels, there always exists a Pareto-optimal solution that activates no more than three constellations. Consequently, without loss of optimality, the original high-dimensional simplex optimization can be reduced to a finite set of low-cardinality mixtures. Compared with fixed-power analyses in which all subcarriers use a single constellation, our method delivers a standards-compliant heterogeneous mixture policy. For frequency-selective channels we design a bilevel scheme that maintains binary constellation decisions at every iterate. The method attains near-optimal performance at a fraction of the complexity.}}
	
	\item {{We validate the proposed theory and algorithms through simulations and SDR-based over-the-air experiments in real environments. The results confirm that increasing the constant-envelope share suppresses MF sidelobes, coherent integration delivers the predicted gain, and mixed-constellation transmission remains reliably demodulable when the per-subcarrier constellation map is known. The experiments also provide physical insight into the MF-RF comparison: MF is preferable in low-SNR regimes, whereas RF becomes superior in high-SNR or clutter-limited regimes because it mitigates symbol-induced clutter interference through deconvolution of the known transmitted symbols.}}
\end{itemize}

Section~II presents the system, signal, and detection models (MF/RF) and defines ESL. Section~III derives single- and multi-symbol sensing expressions and RF noise analysis. Section~IV develops optimal power rules and the joint constellation and power formulations, including structural results and scalable algorithms. Section~V reports simulations and Pluto-SDR experiments. Section~VI concludes.

{{Lowercase letters denote time-domain quantities, uppercase letters denote frequency-domain quantities, and bracket notation $[\cdot]$ is used for discrete-time/frequency indices. Subscripts are reserved for target, mode, and symbol-block labels.}}

\begin{figure}[t]
	\centering
	\includegraphics[width=7.4cm]{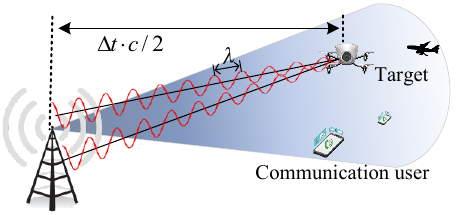}
	\vspace{0mm}
	\caption{Scenarios of ranging measurements using same communication signals.}
	\label{figure1}
\end{figure}

\section{System Model}\label{SystemModel}

We consider a monostatic ISAC system employing an OFDM waveform, and propose a low-complexity {{constellation selection scheme}} to reveal a new trade-off between sensing and communication, while maintaining compatibility with existing communication standards. The ISAC transmitter  emits signals modulated with random communication symbols, which are received by several communication users and simultaneously reflected back to the sensing receiver by $Q$ targets at different ranges, as shown in Fig. \ref{figure1}. The sensing receiver employs classical signal processing techniques such as MF and RF to estimate the range of targets using the known transmitted signals.

\begin{figure}[t]
	\centering
	\includegraphics[width=8.4cm]{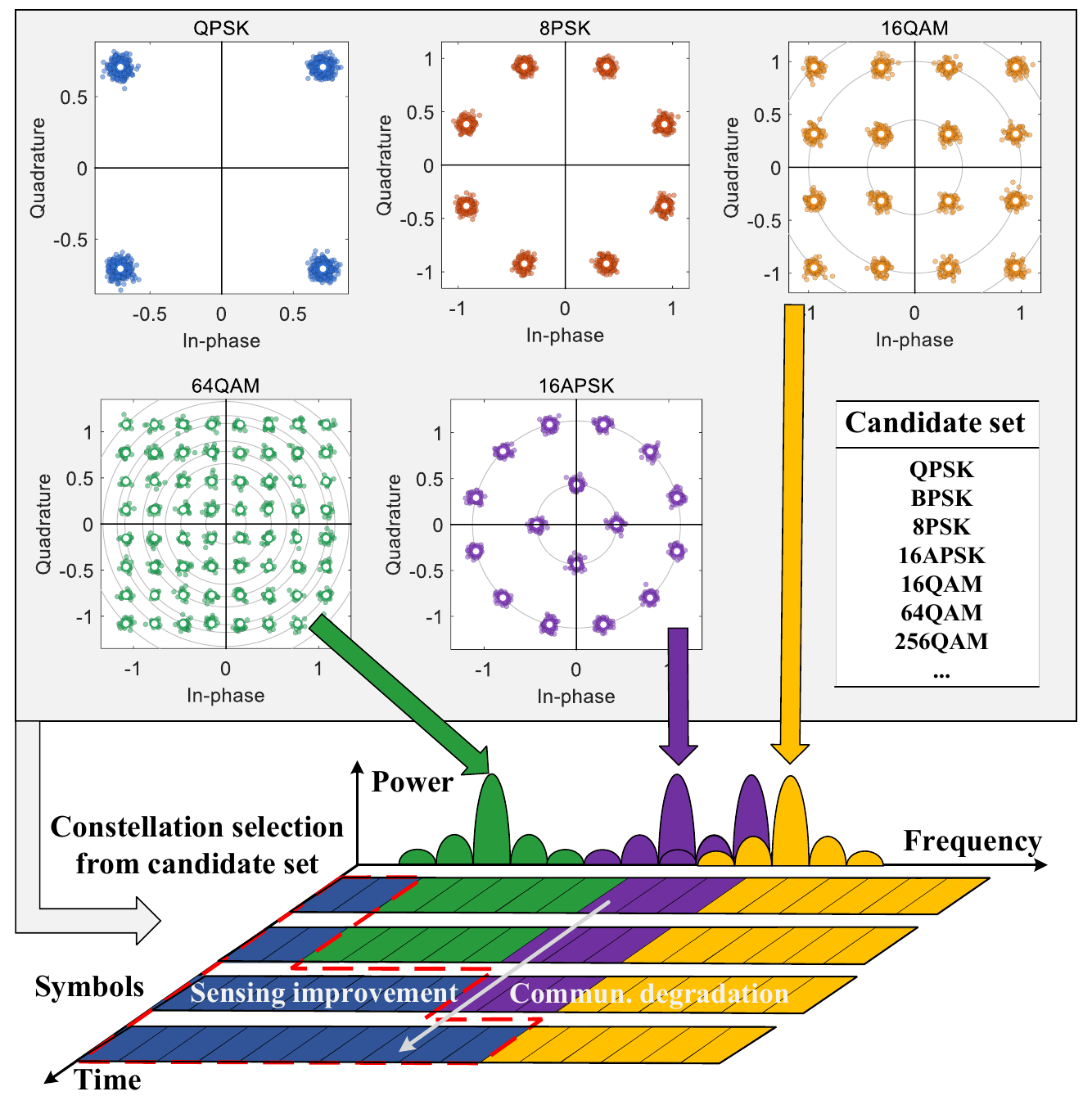}
	\vspace{0mm}
	\caption{Illustration of constellation selection to balance sensing and communication performance.}
	\label{figure1b}
\end{figure}

\subsection{Signal Model}
\label{SignalModel}

{{Let $\mathbf{S} = [S[0], S[1], \ldots, S[N-1]]^T \in \mathbb{C}^{N \times 1}$ denote the transmitted symbol vector over $N$ subcarriers. 
		We consider $J$ candidate modulation constellations, denoted by $\{\mathcal{C}_j\}_{j=1}^J$, where each $\mathcal{C}_j \subset \mathbb{C}$ is a full constellation set (e.g., QPSK, 16QAM, 64QAM, or APSK). 
		Each symbol $S[n]$ is independently drawn from one selected constellation in $\{\mathcal{C}_j\}_{j=1}^J$ according to the sensing and communication requirements, as illustrated in Fig.~\ref{figure1b}.}}\footnote{{Since the choice is restricted to a small finite catalog of standard constellations, the transmitter only needs to signal a compact constellation index to the receiver, analogous to the modulation and coding scheme (MCS) indicator, either per subcarrier or per subcarrier block. In practice, such adaptation is intended for frame-level or scheduling-level updates, rather than symbol-by-symbol signaling, so that the associated control overhead remains small.}}
We assume zero-mean, unit-power symbols, i.e., $\mathbb{E}[S[n]] = 0$ and $\mathbb{E}[|S[n]|^2] = 1$. The kurtosis of the constellation on subcarrier $n$ is defined as \cite{Papoulis2002PRVSP}
\begin{equation} \vspace{0mm}
	\mu_{4,n} \triangleq \frac{\mathbb{E}[|S[n]-\mathbb{E}[S[n]]|^4]}{\mathbb{E}[|S[n]-\mathbb{E}[S[n]]|^2]^2},
\vspace{0mm} \end{equation}
with $\mu_{4,n} \geq 1$. The expected inverse second moment on subcarrier $n$ \cite{11202391} is defined as 
\begin{equation} \vspace{0mm}
	\nu_{-2,n} = \mathbb{E}[ |S[n]|^{-2} ].
\vspace{0mm} \end{equation}
We adopt the unitary $N$-point DFT/IDFT for OFDM modulation. The corresponding time-domain OFDM signal is generated by an $N$-point inverse discrete Fourier transform (IDFT), the discrete-time domain signal is given by 
\begin{equation} \vspace{0mm}
{\bf{x}} = {{\bf F}}^H_N\,{\bf P}\,{\bf S},
\vspace{0mm} \end{equation}
where ${{\bf F}}^H_N$ is the unitary IDFT matrix of size $N$, ${\bf{P}} = {\rm{diag}}\left( \left[ \sqrt{P_0}, \cdots, \sqrt{P_{N-1}} \right] \right)$ denotes the power allocation matrix, and  $P_n$ represents the power allocated to the $n$-th subcarrier. Here, ${\bf{x}} = \left[ x[0], \cdots, x[N-1] \right]$, where
\begin{equation} \vspace{0mm}
	x[t] = \frac{1}{\sqrt{N}} \sum_{n=0}^{N-1} \sqrt{P_n} S[n] e^{j\frac{2\pi}{N} t n}, \quad t = 0,1,\ldots,N-1.
\vspace{0mm} \end{equation}
A cyclic prefix (CP) of appropriate length is appended to mitigate inter-symbol interference, and is discarded at the receiver \cite{goldsmith2005wireless}. Assuming \(Q\) point targets with round-trip propagation delays \(\{\tau_q\}_{q=1}^Q\), the received baseband signal \(\mathbf{y} \in \mathbb{C}^N\) can be modeled as
\begin{equation} \vspace{0mm}
	{\mathbf{y}} = \sum_{q=1}^Q \alpha_q {\mathbf{J}}_{\tau_q} {\mathbf{x}} + {\mathbf{z}},
\vspace{0mm} \end{equation}
where {{\(\alpha_q \in \mathbb{C}\) is the complex reflection coefficient (including pathloss) of the \(q\)-th target, and it is assumed that $\alpha_q \sim \mathcal{CN}({0}, \sigma_{\alpha_q}^2)$.}} Also, \(\mathbf{z} \sim \mathcal{CN}(\mathbf{0}, \sigma_z^2 \mathbf{I})\) is additive white Gaussian noise, and \({\mathbf{J}}_{k}\) is the \(N \times N\) cyclic shift matrix corresponding to delay \(k\) defined as
\begin{equation} \vspace{0mm}
	{\mathbf{J}}_k =
	\begin{bmatrix}
		\mathbf{0}_{k\times (N-k)} & \mathbf{I}_k \\
		\mathbf{I}_{N-k} & \mathbf{0}_{(N-k)\times k}
	\end{bmatrix}, 
	\quad k=0,1,\ldots,N-1 .
\vspace{0mm} \end{equation}
The delay \(\tau_q\) relates to physical range \(d_q\) by
$
	\tau_q = \mathrm{round} \left(  \frac{2 d_q}{c T_s} \right),
$
where \(c\) is the speed of light and \(T_s\) is the sampling interval. This discrete delay simplifies the time-domain representation as cyclic shifts.  

Let \(X[n]\) and \(Y[n]\) be the \(N\)-point DFTs of \(\mathbf{x}\) and \(\mathbf{y}\), respectively, where $X[n] = \sqrt{P_n} S[n]$. 
In a multi-target environment with additive white Gaussian noise (AWGN), the received frequency-domain signal can be expressed as
\begin{equation} \vspace{0mm}\label{RF_Y}
	Y[n] = \sum_{q=1}^Q \alpha_q X[n] e^{-j \frac{2\pi}{N} n \tau_q} + Z[n].
\vspace{0mm} \end{equation}

{{For coherent processing over multiple OFDM symbols, the above single-symbol model extends naturally to the $m$-th OFDM symbol as
		\begin{equation} \vspace{0mm}
			x^{(m)}[t] = \frac{1}{\sqrt{N}}\sum_{n=0}^{N-1}\sqrt{P_n}\, S^{(m)}[n] e^{j\frac{2\pi}{N}nt},
			\vspace{0mm} \end{equation}
		with the corresponding received signal
		\begin{equation} \vspace{0mm}
			\mathbf{y}^{(m)} = \sum_{q=1}^{Q}\alpha_q \mathbf{J}_{\tau_q}\mathbf{x}^{(m)} + \mathbf{z}^{(m)},
			\quad m=1,\ldots,M.
			\vspace{0mm} \end{equation}
		Here, \(\{S^{(m)}[n]\}\) are independent across \(m\) and, for each subcarrier \(n\), follow the same zero-mean, unit-power distribution characterized by \(\mu_{4,n}\) and \(\nu_{-2,n}\) defined above. We assume quasi-static targets over the coherent processing interval, such that \(\{\alpha_q,\tau_q\}_{q=1}^{Q}\) remain approximately unchanged across the \(M\) OFDM symbols.}}

\subsection{Ranging Estimation}\label{RangeEstimation}

\subsubsection{Matched Filtering}

{{For the considered OFDM sensing framework, the MF is implemented in the frequency domain after CP removal and FFT. Let \(X[n]\) and \(Y[n]\) denote the \(N\)-point DFTs of the transmitted and received signals, respectively. The MF applies matched subcarrier-wise weighting as
		\begin{equation} \vspace{0mm}
			\tilde{Y}^{\mathrm{MF}}[n] = \sqrt{N}\,Y[n]X[n]^*, \quad n=0,1,\ldots,N-1.
			\vspace{0mm} \end{equation}
		The corresponding time-domain range profile is then obtained via IDFT:
		\begin{equation} \vspace{0mm}
			\tilde{y}^{\mathrm{MF}}[t] = \frac{1}{\sqrt{N}} \sum_{n=0}^{N-1} \tilde{Y}^{\mathrm{MF}}[n] e^{j \frac{2\pi}{N} n t}.
			\vspace{0mm} \end{equation}
		Equivalently, by the circular-correlation theorem, the same MF output admits the following time-domain interpretation:
		\begin{equation} \vspace{0mm}
			\tilde{y}^{\mathrm{MF}}[t] = \mathbf{x}^H \mathbf{y}_t,
			\vspace{0mm} \end{equation}
		where \(\mathbf{y}_t = [y[t], y[t+1], \ldots, y[t+N-1]]^T\) is the received vector segment starting at sample \(t\), with indices modulo \(N\). Substituting \(\mathbf{y}\), we obtain
		\begin{equation} \vspace{0mm}
			\tilde{y}^{\mathrm{MF}}[t] = \sum_{q=1}^Q \alpha_q \mathbf{x}^H {\mathbf{J}}_{\tau_q - t} \mathbf{x} + \tilde{z}^{\mathrm{MF}}[t],
			\vspace{0mm} \end{equation}
		where \(\tilde z^{\mathrm{MF}}[t]=\mathbf x^H \mathbf J_{N-t}\mathbf z\), which is distributionally equivalent to \(\mathbf x^H\mathbf z\) because \(\mathbf z\sim\mathcal{CN}(\mathbf 0,\sigma_z^2\mathbf I)\) and \(\mathbf J_{N-t}\) is unitary. The key term
		\begin{equation} \vspace{0mm}
			{r}_k = \mathbf{x}^H {\mathbf{J}}_k \mathbf{x} = {r}_{N-k}^*, \quad k=0,1,\ldots,N-1,
			\vspace{0mm} \end{equation}
		defines the periodic autocorrelation functions (ACFs) of the transmit signal \(\mathbf{x}\). Hence, the MF output can be written as a superposition of shifted autocorrelation functions:
		\begin{equation} \vspace{0mm}
			\tilde{y}^{\mathrm{MF}}[t] = \sum_{q=1}^Q \alpha_q {r}_{\tau_q - t} + \tilde{z}^{\mathrm{MF}}[t].
			\vspace{0mm} \end{equation}}}
{Without loss of generality, let $q\in\mathcal{Q} = \{1,\cdots,Q\}$ denote the target of interest, and echoes from all $k\in\mathcal{Q}\setminus\{q\}$ are modeled as clutter.} To quantify the sensing quality under multi-target scenarios and additive noise, we define an effective signal-to-interference-plus-noise ratio (SINR) to reveal the relationship between the signal after receive signal processing and sensing task performance. Specifically, the effective SINR at the true delay $\tau_q$ of the $q$-th target can be defined by
\begin{equation} \vspace{0mm}
	\mathrm{SINR}_q = \frac{\mathbb{E} \left[ \left| s_q^{\mathrm{MF}}[\tau_q] \right|^2 \right]}{\mathbb{E} \left[ \left| \tilde{y}^{\mathrm{MF}}[\tau_q] - s_q^{\mathrm{MF}}[\tau_q] \right|^2 \right]},
\vspace{0mm} \end{equation}
where $\tilde{y}^{\mathrm{MF}}[t]$ is the total output of the sensing filter at delay $t$, and $s_q^{\mathrm{MF}}[t]$ denotes the desired signal contribution from the $q$-th target, i.e., the noise-free component at that location.
Then, $\tilde{y}^{\mathrm{MF}}[t]$ is a superposition of delayed ACFs plus filtered noise, given by
$
\tilde{y}^{\mathrm{MF}}[t] = \sum_{q=1}^{Q} \alpha_q r_{\tau_q - t} + \tilde{z}^{\mathrm{MF}}[t]
$. The desired signal term is $s_q^{\mathrm{MF}}[t] = \alpha_q r_{\tau_q - t}$, while the remaining part of $\tilde{y}^{\mathrm{MF}}[t]$ corresponds to clutter interference from the other targets and filtered noise. Thus
\begin{equation} \vspace{0mm}\label{SINR_Expression_MF}
	\mathrm{SINR}_q^{\mathrm{MF}} = \frac{|\alpha_q|^2 \mathbb{E}[|{r}_0|^2]}{\sum_{p \neq q} |\alpha_p|^2 \mathbb{E}[|{r}_{\tau_q - \tau_p}|^2] + \mathbb{E}[|\tilde{z}^{\mathrm{MF}}[\tau_q]|^2]}.
\vspace{0mm} \end{equation}
\begin{remark}
	{{In (\ref{SINR_Expression_MF}), $\mathbb{E}[|r_0|^2]$ captures the mainlobe energy at the true delay bin, while $\mathbb{E}[|r_{\tau_q-\tau_p}|^2]$ for $p\neq q$ quantifies the masking effect caused by other targets through the off-peak values of the transmit autocorrelation. Therefore, (\ref{SINR_Expression_MF}) should be interpreted as a local output SINR at the $q$-th target bin, rather than a global peak-to-sidelobe ratio over the entire delay axis. 
			The self-sidelobes of the $q$-th target indeed exist, but they appear at off-target delay bins $t\neq \tau_q$ and hence do not enter the denominator of (\ref{SINR_Expression_MF}), which is evaluated specifically at the true delay bin. In the single-target case ($Q=1$), (\ref{SINR_Expression_MF}) therefore reduces to the standard MF output SNR at the matched bin, while the sidelobe behavior is separately characterized by $\mathbb{E}[|r_k|^2]$ for $k\neq 0$.}}
\end{remark}

\subsubsection{Reciprocal Filtering}
{{Similar to the MF above, RF is also implemented in the frequency domain \cite{Wojaczek2019Reciprocal}.}} Given the \(N\)-point DFTs of the transmitted and received signals, denoted by \(X[n]\) and \(Y[n]\), respectively, RF performs element-wise reciprocal weighting as follows:\footnote{We assume $X[n]\neq 0$ almost surely and implement reciprocal filtering via a Tikhonov-regularized division to avoid blow-ups when $|X[n]|$ is small:
	$
	\tilde{Y}^{\mathrm{RF}}[n]=\frac{Y[n]\,X[n]^*}{|X[n]|^2+\varepsilon},\varepsilon>0.
	$ {{Here, $\varepsilon$ is an implementation-level regularization parameter for avoiding numerical instability in reciprocal filtering.}}}
\begin{equation} \vspace{0mm}\label{RFProcess}
	\tilde{Y}^{\mathrm{RF}}[n] = {Y[n]} / {X[n]}, \quad n = 0, \ldots, N-1.
\vspace{0mm} \end{equation}
The time-domain output is then obtained by inverse DFT, i.e.,
\begin{equation} \vspace{0mm}
	\tilde{y}^{\mathrm{RF}}[t] = \frac{1}{\sqrt{N}} \sum_{n=0}^{N-1} \tilde{Y}^{\mathrm{RF}}[n] e^{j \frac{2\pi}{N} n t}.
\vspace{0mm} \end{equation}
Substituting (\ref{RF_Y}) into the RF procedure (\ref{RFProcess}) yields
{{\begin{equation} \vspace{0mm}\label{RFNoise}
			\tilde{y}^{\mathrm{RF}}[t]
			=
			\sum_{q=1}^{Q}\alpha_q
			\underbrace{\frac{1}{\sqrt{N}}\sum_{n=0}^{N-1}e^{j\frac{2\pi}{N}n(t-\tau_q)}}_{\displaystyle \sqrt{N}\,\delta_N[t-\tau_q]}
			+\tilde{z}^{\mathrm{RF}}[t],
			\vspace{0mm} \end{equation}
		where \(\delta_N[\ell]\) denotes the periodic Kronecker delta, i.e.,
		$
		\delta_N[\ell]=
		\begin{cases}
			1, & \ell=0 \ (\mathrm{mod}\ N),\\
			0, & \text{otherwise}.
		\end{cases}
		$
		Therefore, under the adopted integer-delay model, the noise-free RF output consists of impulses located exactly at the corresponding delay bins.
}} In (\ref{RFNoise}), \(\tilde{z}^{\mathrm{RF}}[t]\) is the filtered noise shaped by \(Z[n]/X[n]\). Hence, the RF-based SNR at the \(q\)-th target bin can be written as
\begin{equation} \vspace{0mm}
	\mathrm{SNR}_q^{\mathrm{RF}} = \frac{|\alpha_q|^2 N}{\mathbb{E}[|\tilde{z}^{\mathrm{RF}}[\tau_q]|^2]}.
	\vspace{0mm} \end{equation}
It can be observed that unlike the SINR expression for MF, $\tilde{y}^{\mathrm{RF}}[t]$ does not include sidelobe interference terms in the denominator.
This is because RF performs a symbol-wise division by the known communication symbols, effectively equalizing the frequency-domain response of each subcarrier.
However, the symbol-wise division inherent in RF alters the noise statistics, which can increase the effective noise variance. Specifically, in (\ref{RFNoise}), 
$\tilde{z}^{\mathrm{RF}}[t]$ is the spectrally shaped noise introduced by dividing $Y[n]$ by $X[n]$ in the frequency domain. 

\subsubsection{Matrix Pencil Delay Estimation}
{{Conventional delay estimation based on peak picking of the sensing-filter output is fundamentally limited by the sampling grid. For example, with a bandwidth of $B=20$ MHz, the nominal range resolution is only about $7.5$ m, which motivates the use of super-resolution methods for sub-bin delay estimation.}}
The Matrix Pencil (MP) method \cite{sarkar1995using} enables super-resolution by exploiting the phase progression of complex exponentials embedded in structured data matrices.
{{Specifically, based on the discussion in \ref{RangeEstimation}, let $e[n] = \tilde{Y}^{\mathrm{MF}}[n]$ for matched filtering and  $e[n] = \tilde{Y}^{\mathrm{RF}}[n]$  for reciprocal filtering.}}
Build Hankel matrices of size $L\times K$ with $L=N-K$ and $K\approx N/2$:
\begin{align}
	E_1 &=
	\begin{bmatrix}
		e[0] & e[1] & \cdots & e[K-1]\\
		e[1] & e[2] & \cdots & e[K]\\
		\vdots & \vdots & \ddots & \vdots\\
		e[L-1] & e[L] & \cdots & e[L+K-2]
	\end{bmatrix}, \quad
\end{align}
\begin{align}
	E_2 & =
	\begin{bmatrix}
	e[1] & e[2] & \cdots & e[K]\\
	e[2] & e[3] & \cdots & e[K+1]\\
	\vdots & \vdots & \ddots & \vdots\\
	e[L] & e[L+1] & \cdots & e[L+K-1]
	\end{bmatrix}.\quad
\end{align}
Choose $Q<\min\{K,L\}$ and compute the rank-$Q$ truncated SVD $E_1\approx U\Sigma V^H$ with
$U\in\mathbb C^{L\times Q}$, $\Sigma\in\mathbb R^{Q\times Q}$, $V\in\mathbb C^{K\times Q}$.
Solve
\begin{equation} \vspace{0mm}
	\Sigma^{-1} U^H E_2 V\,\mathbf v_q \;=\; \lambda_q\,\mathbf v_q,
\vspace{0mm} \end{equation}
which ideally yields $\lambda_q = e^{-j\frac{2\pi}{N}\tau_q}$. The delays are then
\begin{equation} \vspace{0mm}
	\hat \tau_q \;=\; -\,\frac{N}{2\pi}\,\arg(\lambda_q) \;\bmod N,\qquad q=1,\ldots,Q.
\vspace{0mm} \end{equation}
{{A higher post-processing output SINR generally leads to a lower delay estimation variance for both super-resolution and peak-based estimators, which explains why the derived MF/RF sensing metrics correlate positively with the ranging accuracy observed in our results \cite{stoica1989music}.}}

\subsection{Communication Model and Performance Metrics}

{{Following communication standards, the system supports downlink communication by embedding information within the frequency-domain symbols $S[n]$. Each subcarrier employs one constellation selected from the candidate set $\{\mathcal{C}_j\}_{j=1}^J$ defined in Section \ref{SignalModel}. For each candidate constellation $\mathcal{C}_j$, its cardinality is denoted by $|\mathcal{C}_j|$, and the corresponding uncoded bits per symbol is $\log_2 |\mathcal{C}_j|$.}}
Let $\chi_{n,j}\in\{0,1\}$ indicate that subcarrier $n$ uses constellation $j\in\{1,\dots,J\}$. Then, the aggregate number of raw bits per OFDM symbol across all users is given by:
\begin{equation} \vspace{0mm}
	R_{\mathrm{com}}=\sum_{n=0}^{N-1}\sum_{j=1}^J \chi_{n,j} \log_2 {|\mathcal C_j|}.
\vspace{0mm} \end{equation}
{{To ensure reliable transmission, the bit error rate (BER) on each subcarrier must remain below a predefined threshold. Accordingly, for each subcarrier \(n\), the BER associated with the selected constellation must satisfy
}}
\begin{equation} \vspace{0mm}
	\sum_{j=1}^J \chi_{n,j} \mathrm{BER}_j\big(\gamma_n\big) \le \mathrm{BER}_{\mathrm{th}}, \forall n,
\vspace{0mm} \end{equation}
{{}{where $ \gamma_n=\frac{|H_n|^2 P_n}{N_0\Delta f}$, and $H_n$ is the complex channel gain on subcarrier $n$. We denote 
$N_0$ as the single-sided noise power spectral density, and $\Delta f$ as subcarrier spacing. Therefore, the communication BER constraint is imposed through the effective per-subcarrier SNR $\gamma_n=|H_n|^2P_n/(N_0\Delta f)$.
More specifically, under AWGN with coherent detection and Gray mapping, BER admits closed forms, e.g., for QPSK $\mathrm{BER}_{\text{QPSK}}(\gamma)=Q(\sqrt{\gamma})$ and for square $M$-QAM $\mathrm{BER}_{\text{QAM}}(\gamma)\approx \tfrac{4}{\log_2 M}\!\left(1-\tfrac{1}{\sqrt{M}}\right)\!Q\!\big(\sqrt{3\gamma/(M-1)}\big)$, where $\gamma$ is the per-subcarrier symbol SNR.
}} Alternatively, we tabulate BER-SNR curves per constellation (including APSK and coded MCS) and apply monotone interpolation/inversion to obtain $\gamma_j^{\min}$ and thus $P_{n,j}^{\min}$ for a given BER target.

\subsection{Problem Formulation}

Communication-centric ISAC entails a communication performance guarantee while maximising sensing performance. Since the optimization ultimately targets sidelobe suppression rather than the performance of any specific target, the formulation does not depend on individual target characteristics. We jointly optimize constellation selection and power allocation over $N$ subcarriers to maximize sensing performance while guaranteeing communication requirements:
\begin{subequations}\label{prob:P1}
	\begin{align}
		(\mathrm{P1})\quad
		\max_{{\bm{\chi}},\mathbf P} \quad 
		& \Gamma_s({\bm{\chi}},\mathbf P)  \label{P1:obj}\\
		\text{s.t.}\quad 
		& \frac{1}{N} R_{\mathrm{com}} \ge R_{\min}, \label{P1:rate}\\
		& \sum_{j=1}^J \chi_{n,j}\,\mathrm{BER}_j\big(\gamma_n\big) \le \mathrm{BER}_{\mathrm{th}}, \forall n, \label{P1:ber}\\
		& \frac{1}{N}	\sum_{n=0}^{N-1} P_n  = P_{\text{ave}},  P_n \ge 0, \forall n, \label{P1:power}\\
		& \sum_{j=1}^{J} \chi_{n,j} = 1, \chi_{n,j}\in\{0,1\}, \forall n,j. \label{P1:select}
	\end{align}
\end{subequations}
where $\Gamma_s({\bm{\chi}},\mathbf P)\in\{\mathrm{SINR}_q^{\mathrm{MF}},\,\mathrm{SNR}_q^{\mathrm{RF}}\}$ denotes the output SINR after sensing receive filtering. Here, (\ref{P1:rate}) and (\ref{P1:ber}) respectively represent communication spectral efficiency constraint and communication BER constraint. (\ref{P1:power}) ensures the total power constraint of all subcarriers. {{Constraint \eqref{P1:select} is a one-hot assignment constraint, which ensures that each subcarrier can select one and only one constellation from the candidate set.}}\footnote{{It is noted that the present communication-centric formulation assumes that all subcarriers are active. For any selected constellation on an active subcarrier, the BER constraint induces a strictly positive power lower bound. Therefore, zero-power active-subcarrier allocation is not considered in this work and is left as a possible extension for future study.}}
		The resulting mixed-integer nonlinear programming (MINLP) problem explicitly captures the sensing-communication tradeoff, as the choice of constellation and power on each subcarrier jointly affects both the communication rate and the sensing SINR of MF (SNR of RF) through the fourth-order moment $\mu_{4,n}$ (the inverse second moment $\nu_{-2,n}$) and the BER expression.

\begin{remark}
	{{The adopted post-filter SINR/SNR in (P1) is a receiver-output-level sensing-quality metric, rather than an estimator-specific or detection-specific objective. It can serve as a fundamental surrogate for multiple sensing tasks, since improved output SINR is generally associated with lower range estimation error and higher detection probability under standard estimation and detection settings.}}
\end{remark}

\section{Analytical Study of the Sensing Performance}\label{SensingExpression}

In this section, we systematically investigate the closed-form expressions of sensing SINR output as a function of the employed constellation and the allocated power, in the considered ISAC system.

\subsection{SINR Expression for Matched Filtering}
Since \(\mathbf{x}\) depends on the random transmitted symbols \(\{S[n]\}\), \({r}_k\) is a random process. The average system performance of MF is thus highly determined via the expected autocorrelation power
$
\mathbb{E}[|{r}_k|^2],
$
which provides insight into average sidelobe levels and range resolution \cite{Liu2025Uncovering}. 
In (\ref{SINR_Expression_MF}), 
the relative delays 
\(\Delta_{qi} \triangleq (\tau_q - \tau_i) \bmod N\) 
are unknown and can be modeled as random, approximately uniform over 
\(\{1, \dots, N-1\}\). 
Hence, the expected autocorrelation power over nonzero delays is
{{\begin{equation} \vspace{0mm}
	\mathbb{E}_{\Delta}\!\Big[\,\mathbb{E}\left[|r_{\Delta}|^2\right]\,\Big]
	= \frac{1}{N-1} \sum_{k=1}^{N-1} \mathbb{E}[|r_k|^2].
	\vspace{0mm} \end{equation}}}
To this end, the expected clutter interference caused by sidelobe of ACFs can be given by $I_q \approx \big(\sum_{i\neq q}|\alpha_i|^2\big)\,\mathrm{ESL}$.
Then, we define an important component, the \emph{expected sidelobe level} (ESL), given by
\begin{equation} \vspace{0mm}
	\mathrm{ESL} = \frac{1}{N-1} \sum_{k=1}^{N-1} \mathbb{E}\left[ |r_k|^2 \right].
\vspace{0mm} \end{equation}
Physically, the ESL can be interpreted as the level of matched-filter output interference at the target location caused by other interfering scatterers in the presence of clutter.
We next analyze the sensing performance expressions for both single-symbol and multi-symbol cases, highlighting the coherent integration gain achievable in each scenario.

\subsubsection{Single OFDM Symbol}\label{SingleSymbol}

We first analyze the sensing SINR for the case of a single OFDM symbol. For the zero-lag term $k=0$, we have
$
r_0=\sum_{n=0}^{N-1} P_n |S[n]|^2,
$ denoting the peak of the auto-correlation results.
Using independence across subcarriers and $\mathbb{E}[|S[n]|^2]=1$, it follows that
\begin{equation} \vspace{0mm}
	\begin{aligned}
		\mathbb{E}\!\left[|r_0|^2\right]
		&\!=\!\! \sum_{n=0}^{N-1} P_n^2\,\mathbb{E}\!\left[|S[n]|^4\right]
		\!+\!\!\! \sum_{\substack{n,n'=0\\ n\neq n'}}^{N-1}\! \! P_n P_{n'}\mathbb{E}\!\left[|S[n]|^2\right] \! \mathbb{E}\!\left[|S[n']|^2\right] \\
		&= \sum_{n=0}^{N-1} P_n^2\,(\mu_{4,n}-1) + \Bigl(\sum_{n=0}^{N-1} P_n\Bigr)^2.
	\end{aligned}
	\label{eq:E_r0_general}
\vspace{0mm} \end{equation}
In (\ref{eq:E_r0_general}),  under equal (or not-too-peaky) power allocation, the relative contribution of the first term is very small for large $N$. In particular, for constant-modulus constellations (e.g., PSK) where $\mu_{4,n} = 1$, the first term vanishes exactly. Therefore, we simplify $\mathbb{E}\left[|r_0|^2\right]$ as $\Bigl(\sum_{n=0}^{N-1} P_n\Bigr)^2$ in the following optimization design.

\begin{Pro}\label{EISLexpression}
	The expected sidelobe level can be given by
	\begin{equation} \vspace{0mm}\label{ESLexpression}
			\mathrm{ESL} 
			=  \frac{1}{N-1} \left( \sum_{n=0}^{N-1}[(N-1)\mu_{4,n} + 1]P_n^2
			\;-\;\Bigl(\sum_{n=0}^{N-1}P_n\Bigr)^2\right).
	\vspace{0mm} \end{equation}
\end{Pro}
\begin{proof}
Please refer to Appendix A.
\end{proof}

Based on Proposition \ref{EISLexpression}, the SINR of sensing receiver signal after receive filtering can be given by
\begin{equation} \vspace{0mm}\label{SINRMFexpression}
	\mathrm{SINR}_q^{\mathrm{MF}}=\frac{|\alpha_q|^2 \mathbb{E}[|{r}_0|^2]}{\sum_{i \neq q} |\alpha_i|^2 {\mathrm{ESL}}+{\sigma_z^2}\sum_{n=0}^{N-1}P_n}.
\vspace{0mm} \end{equation}
Based on (\ref{SINRMFexpression}), the SINR is maximized when the kurtosis equals 1.  Therefore, to enhance performance, one should select as many subcarriers as possible with kurtosis equal to 1, such as those modulated using constant-modulus schemes like QPSK. As evident from (\ref{ESLexpression}), the ESL is invariant to any permutation of the subcarrier constellation order, enabling us to exploit this property for sensing and communication optimization. {{Importantly, since the reflection coefficient $\alpha_q$ is modeled as a zero-mean complex Gaussian random variable, its power $|\alpha_q|^2$ follows an exponential distribution with mean $\mathbb{E}[|\alpha_q|^2]=\sigma_{\alpha_q}^2$. In the following performance analysis, we therefore consider the average SINR with respect to the target reflection coefficient. Consequently, $|\alpha_q|^2$ can be replaced by its expected value $\sigma_{\alpha_q}^2$. 
		Moreover, since this term appears only as a multiplicative constant that is independent of the optimization variables $({\bm{\chi}},\mathbf P)$, it does not affect the maximization of the objective function and can therefore be omitted in the optimization.}}

\subsubsection{Multi-Symbol Coherent Processing Gain}

{{Building on the multi-symbol signal model introduced in Section II, we now analyze the coherent processing gain obtained by averaging over \(M\) independent OFDM symbols within a short interval where the target responses remain approximately stationary.}}
The averaged ACF at lag $k$ over $M$ symbols is defined as
\begin{equation} \vspace{0mm}
	\overline{r}_k = \frac{1}{M}\sum_{m=1}^{M} {r}^m_k  =\frac{1}{M}\sum_{m=1}^{M}\sum_{n=0}^{N-1}P_n\, |S^{(m)}[n]|^2 \, e^{j\frac{2\pi}{N} n k},
	\vspace{0mm} \end{equation}
where $ {r}^m_k$ denotes the ACF of the $m$th realization. 
\begin{theorem}\label{TheoremSINRcoherent}
	For coherent processing of $M$ independent OFDM symbols, the SINR at the MF output for the $q$-th target is
	\begin{equation} \vspace{0mm}\label{SINRMFexpressionMultiframe}
		\mathrm{SINR}_q^{\mathrm{MF}} = 
		\frac{\sigma_{\alpha_q}^2 \mathbb{E}[|\overline{r}_0|^2]}
		{\frac{\sum_{i \neq q} \sigma_{\alpha_i}^2}{N-1}\sum_{k=1}^{N-1}\mathbb{E}[|\overline{r}_k|^2] 
			+ \frac{\sigma_z^2}{M}\sum_{n=0}^{N-1}P_n},
	\vspace{0mm} \end{equation}
	where
	$
	\mathbb{E}[|\overline{r}_0|^2] 
	= \frac{1}{M}\sum_{n=0}^{N-1}P_n^2(\mu_{4,n}-1) + \left(\sum_{n=0}^{N-1}P_n\right)^2,
	$
	and
	$
	\sum_{k=1}^{N-1} \mathbb{E}[|\overline{r}_k|^2] 
	= \frac{N-1}{M}\sum_{n=0}^{N-1}P_n^2(\mu_{4,n}-1) + N\sum_{n=0}^{N-1}P_n^2 - \left(\sum_{n=0}^{N-1}P_n\right)^2.
	$
\end{theorem}
\begin{proof}
	Please refer to Appendix B.
\end{proof}

These expressions explicitly demonstrate how coherent processing (increasing $M$) reduces random sidelobes and noise effects, leaving structural sidelobes unchanged, thus guiding effective waveform optimization.

\subsection{SNR Analysis for Reciprocal Filtering}
To quantify the performance of reciprocal filtering, we derive its theoretical SNR by analyzing the noise statistics after symbol-wise division in the frequency domain.

\begin{theorem}\label{TheoremSINRcoherentRF}
	For coherent processing over $M$ independent OFDM symbols, the theoretical SNR for reciprocal filtering is
	\begin{equation} \vspace{0mm}
		\mathrm{SNR}_q^{\mathrm{RF}} = \frac{\sigma_{\alpha_q}^2 M N^2}{ {\sigma_z^2}{} \sum_{n=0}^{N-1} \frac{\nu_{-2,n}}{P_n} }.
	\vspace{0mm} \end{equation}
\end{theorem}

\begin{proof}
		Please refer to Appendix C.
\end{proof}

From the above analysis, it is evident that increasing the number of coherently processed symbols leads to a monotonic improvement in effective SNR.  
This enhancement significantly improves ranging accuracy when detecting quasi-static targets.  
However, for non-static targets, the number of coherent symbols should be carefully selected based on the target's velocity to avoid additional interference caused by target movement.

\section{Proposed Algorithm Design and Performance Analysis}\label{ProposedAlgorithm}

\subsection{Sensing-Optimal Performance}
When communication performance requirements are disregarded, the problem reduces to maximizing sensing performance alone.  
It is straightforward to observe that when the modulation constellation satisfies \(\mu_{4,n} = 1\) and \(\nu_{-2,n} = 1\), both interference term of MF and noise term of RF are minimized.  
In the general case, where the constellation does not have these ideal properties, the optimal power allocation can be analyzed as follows.

Based on the analysis in Section \ref{SingleSymbol}, we can simplify the expression of SINR of MF results, and thus minimize the ESL instead of maximizing SINR, i.e., 
$
\tilde{\mathrm{SINR}}_q^{\mathrm{MF}} \approx 
\frac{\sigma_{\alpha_q}^2 N^2 P_{\text{ave}}^2}
{\frac{\sum_{i \neq q} \sigma_{\alpha_i}^2}{N-1}\left(\frac{N-1}{M}\sum_{n=0}^{N-1}P_n^2(\mu_{4,n}-1) + N\sum_{n=0}^{N-1}P_n^2 - N^2 P_{\text{ave}}^2\right)
	+ \frac{\sigma_z^2}{M} N P_{\text{ave}}}
$. 

\begin{table}[!t]
\centering
\footnotesize
\caption{Unified Processing and Power-Allocation Rules}
\label{table1}
{{\begin{tabular}{|p{1.8cm}|c|c|}
	\hline
	\centering\textbf{Scheme} 
	& \textbf{Processing}
	& \textbf{Optimal power allocation} \\ \hline
	\makecell[c]{Matched\\filtering}
	&  $\tilde{Y}^{\mathrm{MF}}[n] \leftarrow Y[n]X[n]^*$
	& \makecell{
		$P_n = \dfrac{N P_{\text{ave}}}{b_n \sum_{i=0}^{N-1} b_i^{-1}},$ \\ 
		$b_n =\!  \dfrac{\mu_{4,n} \! -\!  1}{M} \! + \!  \dfrac{N}{N\! -\! 1}$
	} \\ \hline
	\makecell[c]{Reciprocal\\filtering}
	& $\tilde{Y}^{\mathrm{RF}}[n] \leftarrow Y[n] / X[n]$
	& $P_n = \dfrac{N P_{\text{ave}} \sqrt{\nu_{-2,n}}}{\sum_{i=0}^{N-1} \sqrt{\nu_{-2,i}}}$ \\ \hline
\end{tabular}}}
\end{table}

\begin{Pro}\label{thm:eqpower_subset_opt}
	Let $\mathcal{I}\subseteq\{0,\ldots,N-1\}$ with {{$\nu_{-2,n}=\nu_{-2}$}} and $\mu_{4,n}=\mu_4$ for all $n\in \mathcal{I}$.
	Then the optimal allocation satisfies $P_n=P_{n'}$ for all $n,n'\in \mathcal{I}$.
\end{Pro}

\begin{proof}
	Let $P_{\mathcal{I}} = \sum_{n\in \mathcal{I}} P_n$.  
	For MF, the problem can be equivalently transformed into minimizing $\sum_{n\in \mathcal{I}} P_n^2$ s.t. $\sum_{n\in \mathcal{I}} P_n=P_{\mathcal{I}}$.
	Since $x^2$ is convex,
	\begin{equation} \vspace{0mm}
		\sum_{n\in \mathcal{I}} P_n^2 \ \ge\ |\mathcal{I}|\Big(\tfrac{P_{\mathcal{I}}}{|\mathcal{I}|}\Big)^2=\tfrac{P_{\mathcal{I}}^2}{|\mathcal{I}|},
		\vspace{0mm} \end{equation}
	with equality if and only if $P_n=P_{\mathcal{I}}/|\mathcal{I}|$ for all $n\in \mathcal{I}$.
	For RF, the denominator depends on $\sum_{n\in \mathcal{I}} 1/P_n$.
	Similarly, we have,
	$
	\sum_{n\in \mathcal{I}} \tfrac{1}{P_n} \ \ge\ |\mathcal{I}|\Big(\tfrac{P_{\mathcal{I}}}{|\mathcal{I}|}\Big)^{-1}=\tfrac{|\mathcal{I}|^2}{P_{\mathcal{I}}},
	$
	with equality if and only if $P_n=P_{\mathcal{I}}/|\mathcal{I}|$ for all $n\in \mathcal{I}$.
	In both cases the SINR numerator is unchanged, hence equal power on $\mathcal{I}$ is optimal.
\end{proof}

\subsubsection{Matched Filtering}

Then, the optimal power allocation problem can be re-written as
\begin{equation} \vspace{0mm}
	\begin{aligned}
		\min_{\{P_n \geq 0\}} \quad 
		\sum_{n=0}^{N-1} b_n P_n^2, \quad \text{s.t.} \quad \sum_{n=0}^{N-1} P_n = N P_\text{ave}.
	\end{aligned}
\vspace{0mm} \end{equation}
where $ b_n = \frac{1}{M} \left(\mu_{4,n} - 1 \right) + \frac{N}{N-1}$. The Lagrangian is
\begin{equation} \vspace{0mm}
	L = \sum_{n=0}^{N-1} b_n P_n^2 - \lambda \left(  \sum_{n=0}^{N-1} P_n - N P_\text{ave} \right).
\vspace{0mm} \end{equation}
Taking the derivative with respect to $P_n$ and setting it to zero yields
$
2 b_n P_n - \lambda = 0$, $P_n^* = \frac{\lambda}{2 b_n}
$.
Applying the power constraint,
$
\sum_{n=0}^{N-1} \frac{\lambda}{2 b_n} = N P_\text{ave}$, 
$\lambda = \frac{2 N P_\text{ave}}{\sum_{n=0}^{N-1} b_n^{-1}}
$.
Therefore, the optimal allocation is
$
P_i^* = \frac{N P_\text{ave}}{b_i \sum_{n=0}^{N-1} b_n^{-1}}
$.
Subcarriers with higher fourth moments, i.e., larger kurtosis, receive less power.

\subsubsection{Reciprocal Filtering}

For the reciprocal filter, the optimal power allocation is the solution to
\begin{equation} \vspace{0mm}
	\begin{aligned}
		\max_{\{P_n \geq 0\}} \quad 
		\frac{ \left( \sum_{n=0}^{N-1} P_n \right)^2 }
		{ \sum_{n=0}^{N-1} \nu_{-2,n}/P_n } \quad
		\text{s.t.} \quad &\sum_{n=0}^{N-1} P_n = N P_\text{ave}.
	\end{aligned}
\vspace{0mm} \end{equation}

Using the Cauchy-Schwarz inequality, the minimum of $ \sum_{n=0}^{N-1} \nu_{-2,i}/P_n $ under a sum constraint is achieved when
\begin{equation} \vspace{0mm}
	P_i^* = \frac{N P_\text{ave} \sqrt{\nu_{-2,i}}}{\sum_{n=0}^{N-1} \sqrt{\nu_{-2,n}}}.
\vspace{0mm} \end{equation}
Therefore, subcarriers with larger $\nu_{-2,i}$, i.e., more susceptible to noise amplification in RF, receive more power to compensate.

This section  implies a power allocation strategy for maximizing sensing performance. As summarized in Table \ref{table1}, the MF solution penalizes subcarriers with high kurtosis, whereas the RF solution allocates more power to subcarriers with a larger inverse second moment. These results yield closed-form guidelines for optimal power allocation in ISAC systems with arbitrary modulation constellations.

\vspace{0mm}
\subsection{Joint Constellation Selection and Power Allocation for Sensing Improvement under Communication Guarantees}

In this section, we analyze the joint constellation and power allocation under flat fading and frequency-selective channels. While the former admits a simple solution with at most three active constellations, the latter typically demands more diverse modulation across subcarriers.

\subsubsection{Flat Fading Channel}
We consider a scenario with a sufficiently large number of subcarriers, where the allocation of constellations for each subcarrier can be approximated by a continuous distribution over a set of candidate constellations. Let $\eta_j \in [0, 1]$ denote the proportion of subcarriers assigned to the $j$th candidate constellation. In the following, we rigorously prove that under identical channel gains (i.e., in a flat fading channel), the optimal solution requires at most three nonzero values of $\eta_j$. The flat fading assumption is valid for narrowband systems operating in environments with small delay spread.

By Proposition \ref{thm:eqpower_subset_opt}, all subcarriers using the same constellation are assigned equal power in the optimal solution. Then, when $N \to \infty$, it can be readily proved that the problem of (P1) can be equivalently transformed into
\begin{subequations}\label{eq:constraints}
	\begin{align}
		&\max_{{\boldsymbol{\eta}}, \tilde{\bf{P}}} \quad  \Gamma_s, \\
		\text{s.t.} \quad 
		&\sum_{j=1}^J \eta_j = 1, \quad \eta_j \geq 0,  \forall j, \label{eq:allocation-sum} \\
		& \sum_{j=1}^J \eta_j P_j = {P}_{\text{ave}}, \quad P_j\!\ge P_{j,\min}, \forall j, \label{eq:power-equality} \\ 
		& \sum_{j=1}^J \eta_j R_j \geq R_{\min}, \label{eq:rateconst1}
	\end{align}
\end{subequations}
where ${\boldsymbol{\eta}} = [\eta_1, \cdots, \eta_J]$, $\tilde {\bm{P}} = [P_1, \cdots, P_J]$, $\Gamma_s = \frac{\sigma_{\alpha_q}^2 M N^2}{ {\sigma_z^2}{} N\sum_{j=1}^{J} \eta_j \frac{\nu_{-2,j}}{P_j} }$ for RF, and $\Gamma_s = \frac{
	\sigma_{\alpha_q}^2 N^2  P_{\text{ave}}^2
}{
	{\sum_{i \neq q} \sigma_{\alpha_i}^2} \left[
	\frac{N}{M}{\sum_{j=1}^{J} \eta_j P_j^2 (\mu_{4,j}-1)}
	+ {\frac{N^2\sum_{j=1}^{J} \eta_j P_j^2 - N^2 P_{\text{ave}}^2}{N-1}
		\vphantom{\sum_{n=0}^{N-1} P_n^2}
	}
	\right]
	+ \frac{\sigma_z^2}{M} N P_{\text{ave}}
}$ for MF. Also, $R_j = \log_2|\mathcal{C}_j|$.

To efficiently solve problem (\ref{eq:constraints}), we define the average power share of constellation $j$ as
$
\theta_j \triangleq \eta_j P_j .
$
Then $P_j=\theta_j/\eta_j$ for any $\eta_j>0$, and if $\eta_j=0$ we set $\theta_j=0$
(the choice of $P_j\!\ge0$ is immaterial). With this substitution the constraints
of \eqref{eq:constraints} become
\begin{equation} \vspace{0mm}
	\begin{aligned}
		&\sum_{j=1}^J \eta_j = 1,\; \eta_j\ge0,\qquad
		\sum_{j=1}^J \theta_j =  P_{\text{ave}},\; \theta_j\ge0,\\
		&\sum_{j=1}^J \eta_j R_j \ge R_{\min},\qquad
		\eta_j P_{j,\min} \le \theta_j \le \eta_j P_{\max}, \forall j .
	\end{aligned}
	\label{eq:constr-eta-y}
\vspace{0mm} \end{equation}
Here $P_{\max}>0$ is a sufficiently large constant,
whose sole purpose is to enforce $\eta_j=0 \Rightarrow \theta_j=0$.

For MF cases, noting
$\sum_j \eta_j P_j^2=\sum_j \frac{\theta_j^2}{\eta_j}$ and
$N^2\sum_j \eta_j P_j^2-N^2P_{\text{ave}}^2
= N^2\sum_j \frac{\theta_j^2}{\eta_j}-N^2P_{\text{ave}}^2$,
the MF SINR denominator splits into a constant
plus a variable part. Dropping constants (which do not affect the maximizer),
maximizing the MF $\Gamma_s$ is equivalent to
\begin{equation} \vspace{0mm}
	\min_{\boldsymbol\eta,\bm{\theta}}\;
	\sum_{j=1}^J c_j\,\frac{\theta_j^2}{\eta_j}
	\quad\text{s.t.}\quad \eqref{eq:constr-eta-y},
	\label{eq:P1-MF-eta-y}
\vspace{0mm} \end{equation}
where
$
	c_j = \Big(\sum_{i\neq q}  \sigma_{\alpha_i}^2  \Big)\frac{N}{M}\big(\mu_{4,j}-1\big)
	+ \Big(\sum_{i\neq q}   \sigma_{\alpha_i}^2 \Big)\frac{N^2}{N-1} >0 .
$

For RF cases, since $\sum_j \eta_j \frac{\nu_{-2,j}}{P_j}=\sum_j \nu_{-2,j}\frac{\eta_j^2}{\theta_j}$,
maximizing the RF $\Gamma_s$ is equivalent to the convex program
\begin{equation} \vspace{0mm}
	\min_{\boldsymbol\eta,\bm{\theta}}\;
	 \sum_{j=1}^J \nu_{-2,j}\,\frac{\eta_j^2}{\theta_j}
	\quad\text{s.t.}\quad \eqref{eq:constr-eta-y}.
	\label{eq:P1-RF-eta-y}
\vspace{0mm} \end{equation}

Problems \eqref{eq:P1-MF-eta-y} and \eqref{eq:P1-RF-eta-y} are
\emph{equivalent} to \eqref{eq:constraints} under continuous power
(with BER thresholds handled as $\theta_j\!\ge\!\eta_j P_{j,\min}$).
Both objective functions are convex, as each term $\eta_j^2/\theta_j$ or $\theta_j^2/\eta_j$ is
quadratic-over-linear and the constraints are linear.

{{To further reduce the computational complexity, the following structural results are developed.}}
\begin{theorem}\label{TwoconstellationIsEnough}
	In (\ref{eq:constraints}), there exists an optimal solution $\boldsymbol{\eta}^*$ with at most three active constellation, i.e.,
	\begin{equation} \vspace{0mm}
		\left|\{j: \eta_j^* > 0\}\right| \leq 3.
		\vspace{0mm} \end{equation}
\end{theorem}
\begin{proof}
	Please refer to Appendix D.
\end{proof}
\begin{remark}
	The above result shows that the Pareto-optimal operating point can almost always be achieved by mixing at most three constellations. {{Intuitively, since the solution is a convex combination of constellation types, Carathéodory's theorem \cite{boyd2004convex} implies any optimum uses at most three extreme points, and if the rate constraint is inactive, at most two.}} 
	This has significant engineering value: Instead of an exhaustive search over the entire $J$-dimensional simplex of constellation mixtures, it suffices to consider all $\binom{J}{3}$ possible triples and optimize their respective weights.  The joint sensing-communication trade-off can therefore be efficiently implemented by continuously interpolating among three selected modulation formats. 
\end{remark}

{{By Theorem \ref{TwoconstellationIsEnough}, for each fixed triplet, the remaining optimization can be further reduced to a one-dimensional exact subproblem. For a fixed active triplet $\mathcal S=\{1,2,3\}$ with $R_1\le R_2\le R_3$, and in the generic Pareto-optimal case where the rate constraint is active, the mixture weights can be reduced to one degree of freedom by
		$
		\eta_1+\eta_2+\eta_3=1,$ and $ 
		\eta_1R_1+\eta_2R_2+\eta_3R_3=R_{\min}.
		$
		Specifically, letting $\eta_2=x$ gives
		$
		\eta_1(x)=\frac{R_3-R_{\min}-x(R_3-R_2)}{R_3-R_1},
		\eta_3(x)=\frac{R_{\min}-R_1-x(R_2-R_1)}{R_3-R_1},
		$
		so the reduced problem depends only on the scalar $x$.
		For MF, substituting $\eta_j(x)$ into the equivalent flat-fading formulation yields the inner power-allocation subproblem
		\begin{equation}
			\min_{P_j\ge P_{j,\min}} \;\sum_{j=1}^{3}\eta_j(x)c_j P_j^2
			\quad
			\text{s.t.}\quad
			\sum_{j=1}^{3}\eta_j(x)P_j=P_{\text{ave}},
		\end{equation}
		where
		$
		c_j \triangleq \Big(\sum_{i\neq q}\sigma_{\alpha_i}^2\Big)\frac{N}{M}(\mu_{4,j}-1)
		+\Big(\sum_{i\neq q}\sigma_{\alpha_i}^2\Big)\frac{N^2}{N-1}>0.
		$
		For each feasible BER-active pattern, define
		$
		\mathcal A \triangleq \{\,j\in\mathcal S: P_j^\star=P_{j,\min}\,\},
		\mathcal U \triangleq \mathcal S\setminus\mathcal A.
		$ By the KKT conditions, the optimal powers admit the closed form
		\[
		P_j^\star(x)=
		\begin{cases}
			P_{j,\min}, & j\in\mathcal A,\\[1mm]
			\dfrac{t_{\mathcal A}(x)}{c_j}, & j\in\mathcal U,
		\end{cases}
		\]
		where $t_{\mathcal A}(x)=
		\frac{P_{\text{ave}}-\sum_{j\in\mathcal A}\eta_j(x)P_{j,\min}}
		{\sum_{j\in\mathcal U}\eta_j(x)/c_j}.$
		Substituting $P_j^\star(x)$ back gives the one-dimensional piecewise objective
		\begin{equation}
			f_{\mathcal A}^{\mathrm{MF}}(x)
			=
			\sum_{j\in\mathcal A}\eta_j(x)c_jP_{j,\min}^2
			+
			\frac{\Big(P_{\text{ave}}-\sum_{j\in\mathcal A}\eta_j(x)P_{j,\min}\Big)^2}
			{\sum_{j\in\mathcal U}\eta_j(x)/c_j},
		\end{equation}
		which is of the form
		$
		f_{\mathcal A}^{\mathrm{MF}}(x)=a_0+a_1x+\frac{(b_0+b_1x)^2}{d_0+d_1x}.
		$
		Hence, in each piece, the stationary points are obtained from a quadratic equation, and the exact optimum on the fixed triplet support is found by checking finitely many closed-form candidates together with the interval endpoints. The RF case can be handled in the same manner as above, and the details are omitted for brevity. Therefore, the flat-fading problem can be solved exactly without a generic convex solver by enumerating all feasible triplets.}}\footnote{{The interval endpoints also cover degenerate cases where one mixture weight becomes zero, and therefore include the optimal two-constellation and single-constellation solutions.}}

{{
			For the flat-fading problem, a candidate constellation $j$ can be safely discarded in a preprocessing step if there exists another constellation $i$ such that
			\begin{equation}
				R_i \ge R_j,\qquad
				P_{i,\min}\le P_{j,\min},\qquad
				w_i\le w_j,
			\end{equation}
			where $w_j = \mu_{4,j}$ for MF and $w_j = \nu_{-2,j}$ for RF. 
			Indeed, under these conditions, constellation $i$ is no worse than constellation $j$ in terms of communication rate, BER feasibility, and sensing cost. Therefore, any allocation mass assigned to constellation $j$ can be shifted to constellation $i$ without degrading the objective value, and constellation $j$ need not be considered in the flat-fading optimization. This preprocessing rule helps reduce both the triplet search space and the size of the equivalent convex program.
}}

\subsubsection{Frequency-Selective Channel}

When subcarrier gains $H_n$ vary across $n\in\{0,\dots,N-1\}$, the three-constellation principle in Theorem \ref{TwoconstellationIsEnough} no longer holds, {{so restricting the design to at most three constellations may incur performance loss.}}  Each subcarrier should choose its constellation according to its local channel quality. Define the power relevant variable
\begin{equation} \vspace{0mm}
	\kappa_{n,j} \triangleq \chi_{n,j}\,P_{n,j},
\vspace{0mm} \end{equation}
so that $P_{n,j}$ is the transmit power when $\chi_{n,j}=1$ (otherwise $P_{n,j}$ is immaterial and $\kappa_{n,j}=0$). For a target BER, denote the per-subcarrier minimum power by
$
P_{n,j}^{\min} \triangleq P_j^{\min}(H_n),
$
and let $P_{\max}$ be a loose upper bound.  For MF-based sensing, the coefficient
$
a_j =\Big(\sum_{i\neq q}\! \sigma_{\alpha_i}^2 \Big)\frac{N}{M}\,(\mu_{4,j}-1)
+\Big(\sum_{i\neq q}\! \sigma_{\alpha_i}^2 \Big)\frac{N^2}{N-1}>0
$
collects the modulation-dependent moments.

Then, maximizing MF-SINR is equivalent to the following mixed-integer quadratic-over-linear (QOL) program:
\begin{subequations}\label{eq:FS-MF-problem-clean}
	\begin{align}
		\min_{\chi,\kappa}\quad&
		\sum_{n=0}^{N-1}\sum_{j=1}^{J} a_j\,\frac{\kappa_{n,j}^2}{\chi_{n,j}} ,
		\label{eq:MF-obj}
		\\
		\text{s.t.}\quad&
		\sum_{j=1}^{J} \chi_{n,j}=1,\qquad \chi_{n,j}\in\{0,1\},\ \forall n,
		\label{eq:assign}
		\\
		&
	\frac{1}{N} \sum_{n=0}^{N-1}\sum_{j=1}^{J} \kappa_{n,j}=P_{\text{ave}},\qquad \kappa_{n,j}\ge 0,
		\label{eq:power-sum}
		\\
		&
		\frac{1}{N}\sum_{n=0}^{N-1}\sum_{j=1}^{J} \chi_{n,j} R_j \;\ge\; R_{\min},
		\label{eq:rate-avg}
		\\
		&
		\chi_{n,j} P_{n,j}^{\min}\le \kappa_{n,j}\le \chi_{n,j} P_{\max},\ \forall n,j .
		\label{eq:ber-lb-cap}
	\end{align}
\end{subequations}
The perspective $\kappa_{n,j}^2/\chi_{n,j}$ is jointly convex on $\{\chi_{n,j}>0,\ \kappa_{n,j}\ge0\}$ and enforces $\kappa_{n,j}=0$ if $\chi_{n,j}=0$ via \eqref{eq:ber-lb-cap}. The problem admits an exact reformulation as a mixed-integer second-order cone program (MISOCP). 

We introduce the dual variables $\psi_{\mathcal{L}}\in\mathbb{R}$ for \eqref{eq:power-sum} and $\lambda_{\mathcal{L}}\ge 0$ for \eqref{eq:rate-avg}. Dropping constants, the MF Lagrangian reads
\begin{equation} \vspace{0mm}\label{eq:L-MF}
	\begin{aligned}
		\mathcal{L}_{\mathrm{MF}}(\chi,\kappa;\psi_{\mathcal{L}},\lambda_{\mathcal{L}})
		=&\sum_{n,j} a_j\,\frac{\kappa_{n,j}^2}{\chi_{n,j}}
		-\psi_{\mathcal{L}}\Big(\sum_{n,j} \kappa_{n,j}-N P_{\text{ave}}\Big) \\
		&+\lambda_{\mathcal{L}}\Big(R_{\min}-\tfrac{1}{N}\sum_{n,j} \chi_{n,j} R_j\Big).
	\end{aligned}
\vspace{0mm} \end{equation}
We now design a low-complexity scheme by moving the coupling constraints into the objective via $(\psi_{\mathcal{L}},\lambda_{\mathcal{L}})$ so that, for fixed prices, the $N$ subcarriers decouple.

\paragraph{Inner layer update.}
For a fixed $(\psi_{\mathcal{L}},\lambda_{\mathcal{L}})$ and an active pair $(n,j)$ (i.e., $\chi_{n,j}=1$), minimizing w.r.t. $\kappa_{n,j}$ over the interval $[P_{n,j}^{\min},\,P_{\max}]$ yields the one-dimensional convex quadratic
\begin{equation} \vspace{0mm}
\min_{\,P_{n,j}^{\min}\ \le\ \kappa\ \le\ P_{\max}}\ a_j \kappa^2 - \psi_{\mathcal{L}} \kappa,
\vspace{0mm} \end{equation}
whose optimal value is the reduced cost
\begin{equation} \vspace{0mm}\label{eq:phi}
	\phi_{n,j}(\psi_{\mathcal{L}})
	=
	\begin{cases}
		a_j \big(P_{n,j}^{\min}\big)^2 - \psi_{\mathcal{L}} P_{n,j}^{\min}, 
		& \displaystyle \frac{\psi_{\mathcal{L}}}{2a_j} \le P_{n,j}^{\min},\\[8pt]
		-\dfrac{\psi_{\mathcal{L}}^{2}}{4\,a_j}, 
		& \displaystyle P_{n,j}^{\min} \le \frac{\psi_{\mathcal{L}}}{2a_j} \le P_{\max},\\[10pt]
		a_j P_{\max}^{2} - \psi_{\mathcal{L}} P_{\max}, 
		& \displaystyle \frac{\psi_{\mathcal{L}}}{2a_j} \ge P_{\max}.
	\end{cases}
	\vspace{0mm} \end{equation}
Substituting $\phi_{n,j}(\psi_{\mathcal{L}})$ back into \eqref{eq:L-MF} eliminates $\kappa$ and leaves a linear form in $\chi$:
$
\sum_{n,j} \chi_{n,j}\!\left(\phi_{n,j}(\psi_{\mathcal{L}})-\frac{\lambda_{\mathcal{L}}}{N}R_j\right)
+\psi_{\mathcal{L}} N P_{\text{ave}}+\lambda_{\mathcal{L}} R_{\min}.
$
Due to constraints in \eqref{eq:assign}, the minimum over $\chi$ is attained at a simplex vertex, i.e.,
\begin{equation} \vspace{0mm}\label{eq:select}
	j_n \in \arg\min_{j}\Big\{\,\phi_{n,j}(\psi_{\mathcal{L}})-\tfrac{\lambda_{\mathcal{L}}}{N}R_j\,\Big\},
	\chi_{n,j_n}=1,\ \ \chi_{n,j\neq j_n}=0 .
\vspace{0mm} \end{equation}
Given $\{j_n\}$, under the Lagrangian relaxation the per-subcarrier power update is separable and admits a clamped closed form:
\begin{equation} \vspace{0mm}\label{eq:alloc-closed}
	\kappa_n^\star(\psi_{\mathcal{L}})
	= \min\!\Big\{\,P_{\max},\ \max\!\big\{\,P_{n,j_n}^{\min},\ \tfrac{\psi_{\mathcal{L}}}{2\,a_{j_n}}\,\big\}\Big\}.
	\vspace{0mm} \end{equation}

\paragraph{Outer layer update.}
Let $\overline{R}=\tfrac{1}{N}\sum_{n} R_{j_n}$ be the achieved average rate and $S_y=\sum_n \kappa_n^\star$ the total power returned by \eqref{eq:alloc-closed}. $\psi_{\mathcal{L}}$ and $\lambda_{\mathcal{L}}$ are updated by a projected subgradient, 
\begin{equation} \vspace{0mm}
\psi_{\mathcal{L}}^{(t+1)}=\psi_{\mathcal{L}}^{(t)}+\psi_t\!\left(N P_{\text{ave}}-\sum_{n} \kappa_n^\star\!\big(\psi_{\mathcal{L}}^{(t)}\big)\right)
\vspace{0mm} \end{equation}
and
\begin{equation} \vspace{0mm}\label{eq:lambda-update}
	\lambda_{\mathcal{L}}^{(t+1)}
	= \Big[\, \lambda_{\mathcal{L}}^{(t)} + \psi_t\big(R_{\min}-\overline{R}^{(t)}\big)\,\Big]_+,
\vspace{0mm} \end{equation}
where $\psi_t>0$. 
Combining \eqref{eq:select}-\eqref{eq:lambda-update} yields a scalable bilevel scheme that preserves the binary selection $\chi_{n,j}\in\{0,1\}$, satisfies the BER lower bounds, and computes $\kappa$ in closed form at every iteration.

The RF case uses the same constraint set \eqref{eq:assign}-\eqref{eq:ber-lb-cap} and replaces the objective by
$
\frac{\sigma_z^2}{MN}\sum_{n,j} \nu_{-2,j}\,\frac{\chi_{n,j}^2}{\kappa_{n,j}}.
$
Proceeding identically, first eliminating $\kappa$ to obtain a per-pair reduced cost and then applying the one-hot rule \eqref{eq:select}, leads to an analogous closed-form power update. To save space, we omit the RF algebra; it is entirely similar.

\section{Simulations and Experiments}\label{Experimental}

Using numerical simulation and experimental results, we have studied the fundamental characteristics of constellation combination of ISAC systems. The system parameters are given as follows: the frequency $f^c = 2.45$ GHz, the bandwidth $B = 20$ MHz, subcarrier number $ N = 64$, number of coherent processing symbols $M = 16$. {{The communication reliability constraint is imposed through the modulation-dependent effective-SNR threshold $\gamma_j^{\min}$ required to achieve ${\rm BER}_{\rm th}=10^{-4}$.}}
{{The candidate constellation set consists of QPSK, 16QAM, 64QAM, 256QAM, 8APSK, 16APSK, and 32APSK.}}{\footnote{{Unless otherwise stated, the APSK candidates are normalized ring-APSK alphabets with ring populations/radii given by \(8\mathrm{APSK}:[3,5]/[1,\sqrt{2}]\), \(16\mathrm{APSK}:[4,12]/[1,\sqrt{2}]\), and \(32\mathrm{APSK}:[4,12,16]/[1,\sqrt{2},\sqrt{3}]\). The corresponding fourth-order moments and inverse second moments are \((\mu_4,\nu_{-2})=(1.0888,1.1172)\), \((1.0612,1.0938)\), and \((1.0859,1.1380)\), respectively.}}}

\subsection{Simulations: Trade-offs between Sensing and Communication}

\begin{figure}[t]
	\centering
	\includegraphics[width=8.4cm]{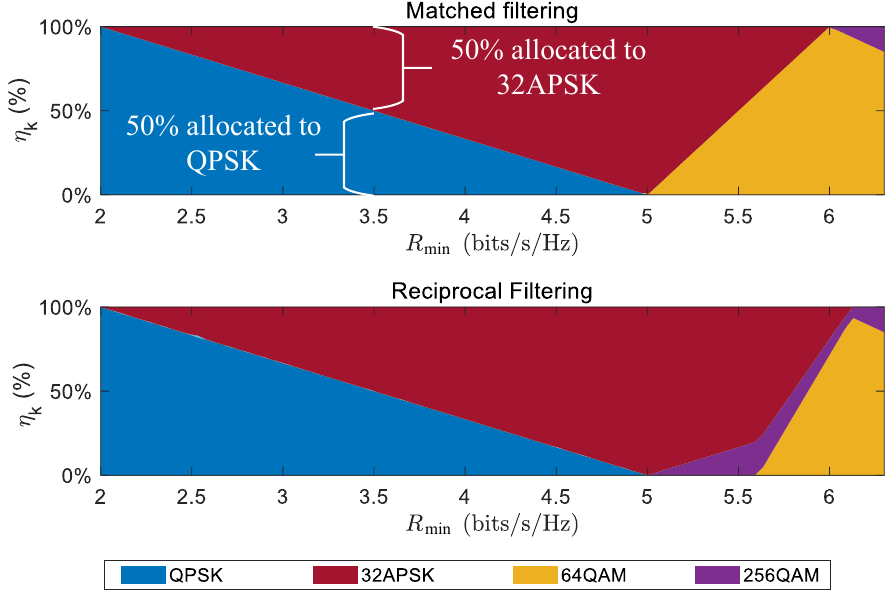}
	\vspace{0mm}
	\caption{{Optimal constellation proportions versus the minimum rate requirement $R_{\min}$ under flat fading ($P_{\text{ave}}=6$).}}
	\label{figure2}
\end{figure}
Fig. \ref{figure2} illustrates the optimal constellation mixture $\boldsymbol{\eta}$ as a function of the minimum rate guarantee $R_{\min}$ for identical subcarrier channel gains, under both MF and RF processing, where the y-axis reports the subcarrier proportion. Specifically, when the rate constraint is 3.5 bits/s/Hz, the optimal constellation mixture for both MF and RF consists of 50\% QPSK and 50\% 32APSK. Consistent with Theorem \ref{TwoconstellationIsEnough} under the flat fading channel, across almost the entire range of $R_{\min}$, the optimizer activates no more than three constellations at a time. For mild rate requirements, the solution concentrates on low-order constellations (predominantly QPSK) to suppress MF sidelobes and reduce the RF noise-variance term. As $R_{\min}$ tightens, the mixture shifts monotonically toward higher-order constellations (e.g., 32APSK/64QAM), reducing the share of low-order formats to meet the rate while maintaining sensing quality. Under frequency-flat channels, MF and RF show almost identical switching thresholds and piecewise-linear transitions, indicating essentially the same constellation mixture.

Fig. \ref{figure2b} shows the optimal joint constellation selection and power allocation over a frequency-selective channel ($N=64$, $M=16$, $R_{\min}=6$ bits/s/Hz). {{For the frequency-selective simulations, the subcarrier gains $\{H_n\}$ are generated from a quasi-static tapped-delay-line channel with exponentially decaying power-delay profile.}} The red curve gives the subcarrier gains $|H_n|^2$ sorted in descending order, while bar colors indicate the chosen constellation and bar heights the allocated power $P_n$. 
For both MF and RF (c.f. Fig. \ref{figure2b}(a) and Fig. \ref{figure2b}(b)), subcarriers with larger $|H_n|^2$ are mapped to higher-order constellations, exploiting the strong channels to meet the communication rate requirements with less power, whereas the subcarriers with weaker channel gain are assigned lower-order constellations to preserve the sensing metric, i.e., reduced MF sidelobes or smaller RF noise variance with only a modest rate penalty. Moreover, within any contiguous block employing the same constellation, the optimal per-subcarrier powers equalize at the BER threshold, in agreement with Proposition~\ref{thm:eqpower_subset_opt}. It is noted that only subcarriers near block boundaries show increased $P_n$ to reach this threshold, after which additional power yields little benefit and is therefore not allocated. The MF and RF solutions share a similar block structure and power allocation trends, differing only by small boundary offsets.

Fig. \ref{figure3} illustrates that range RMSE versus SNR for both MF and RF. In Fig. \ref{figure3}(a), with optimal power allocation and varying constellation mixtures, MF attains lower RMSE at low SNR, whereas RF becomes superior at high SNR. This crossover arises because reciprocal weighting in RF increases effective noise variance at low SNR, especially when the mixture includes constellations without a constant envelope such as 16QAM, while MF is limited at high SNR by a sidelobe floor. A higher fraction of 16QAM consistently degrades sensing for both criteria: RMSE is largest with 100\% 16QAM, intermediate with 50\% QPSK and 50\% 16QAM, and smallest with 100\% QPSK. {{Fig. \ref{figure3}(b) shows that, for a fixed 50\% QPSK and 50\% 16QAM mixture, the ``Optimal'' power allocation given in Table I consistently outperforms both random allocation and the communications water-filling baseline.}} The gain is most pronounced for MF at low SNR and remains substantial for RF at high SNR.

\begin{figure}[t]
	\centering
	\includegraphics[width=8.8cm]{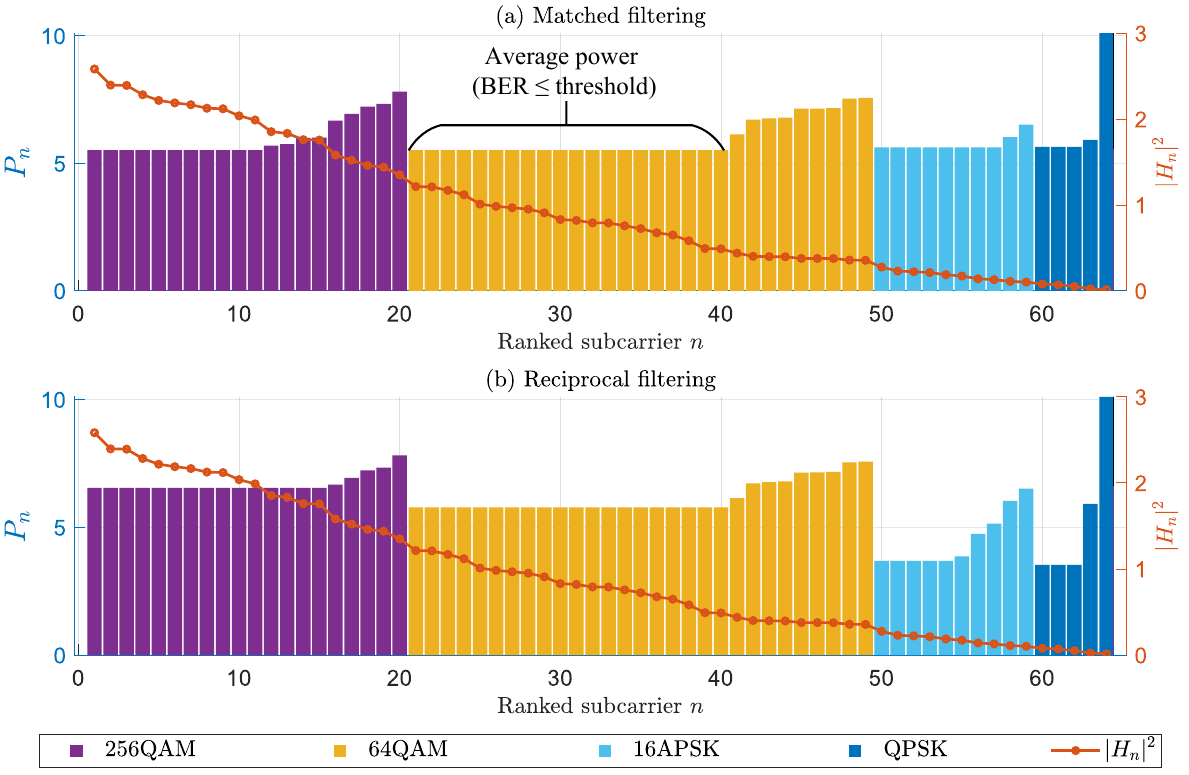}
	\vspace{0mm}
	\caption{Optimal per-subcarrier constellation selection and power allocation in a frequency-selective channel ($R_{\min}=6$ bits/s/Hz, $N=64$, $M=16$).}
	\label{figure2b}
\end{figure}

\begin{figure}[t]
	\centering
	\setlength{\abovecaptionskip}{0.cm}
	\subfigure[Different constellation mixtures with the optimal power allocation specified in Table~I.]
	{	
		\label{figure3a}
		\vspace{0mm}
		\includegraphics[width=8.4cm]{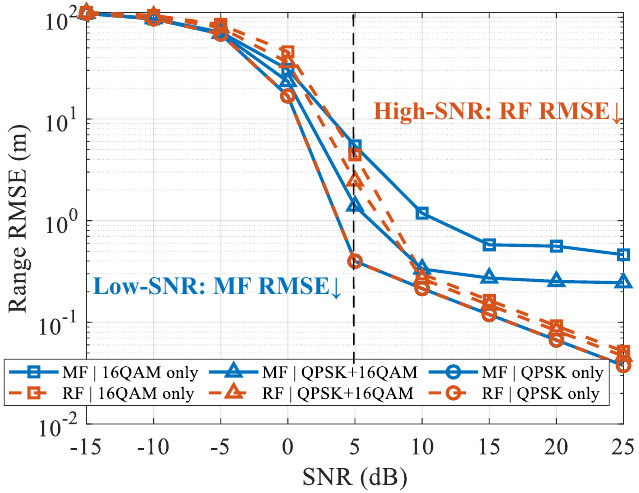}
	}
	\subfigure[Optimal, random, and water-filling power allocation for a fixed 50\% QPSK + 50\% 16QAM mixture. (Waterfilling refers to the optimal power allocation strategy in communication systems based on the water-filling algorithm.)]
	{	
		\label{figure6b}
		\vspace{0mm}
		\includegraphics[width=8.4cm]{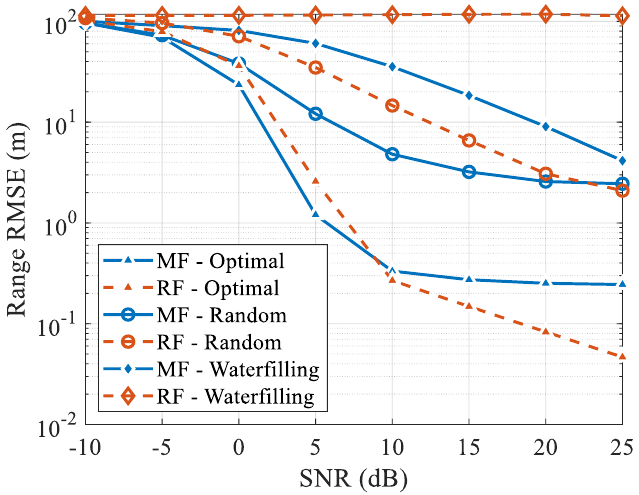}
	}
	\caption{Range RMSE versus SNR for MF and RF under different constellation mixtures and power-allocation strategies without communication constraints.}
	\label{figure3}
\end{figure}

\begin{figure}[t]
	\centering
	\includegraphics[width=8.4cm]{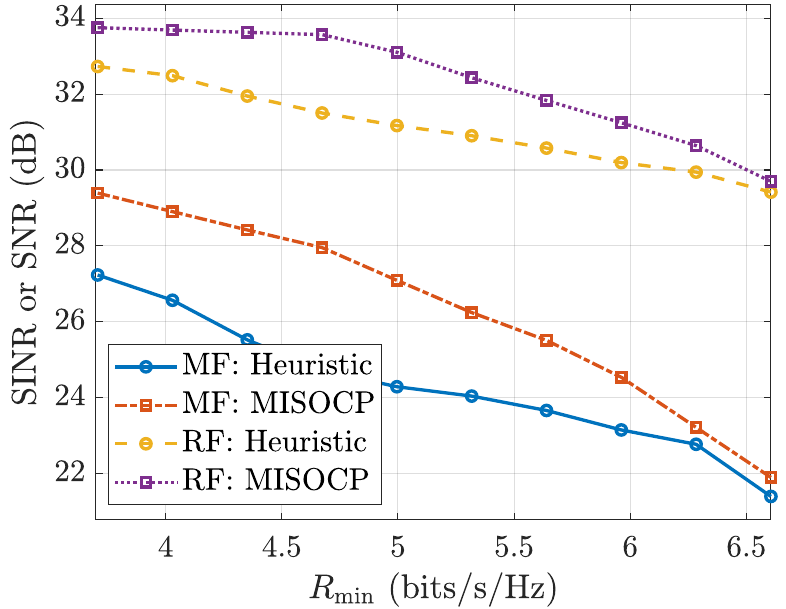}
	\vspace{0mm}
	\caption{Sensing-communication trade-off in (P1) versus the minimum rate requirement $R_{\min}$. {{``Heuristic'' refers to the proposed algorithm in (52)–(58).}}}
	\label{figure4}
\end{figure}
Fig. \ref{figure4} plots the sensing metric versus the minimum rate $R_{\min}$ at fixed transmit power under frequency-selective channel. For MF reported as SINR and RF reported as SNR, the curves decrease monotonically with $R_{\min}$, which quantifies the sensing and communication tradeoff. RF attains a higher metric than MF across the entire range. 
{{The proposed heuristic algorithm (c.f. (\ref{eq:L-MF}) - (\ref{eq:lambda-update}))}} remains close to the MISOCP benchmark with a bounded gap that narrows as $R_{\min}$ increases: about 2 to 3 dB for MF and about 1 to 2 dB for RF at small $R_{\min}$. The residual
gap is mainly due to the integer nature of the constellation
selection variables, the nonzero duality gap introduced by the
dual relaxation, and the fact that the primal solution is recovered
from dual iterates rather than obtained from a globally optimal
joint MISOCP solve.

\subsection{Over-the-air Experiment Results}

\begin{figure}[t]
	\centering
	\includegraphics[width=8.4cm]{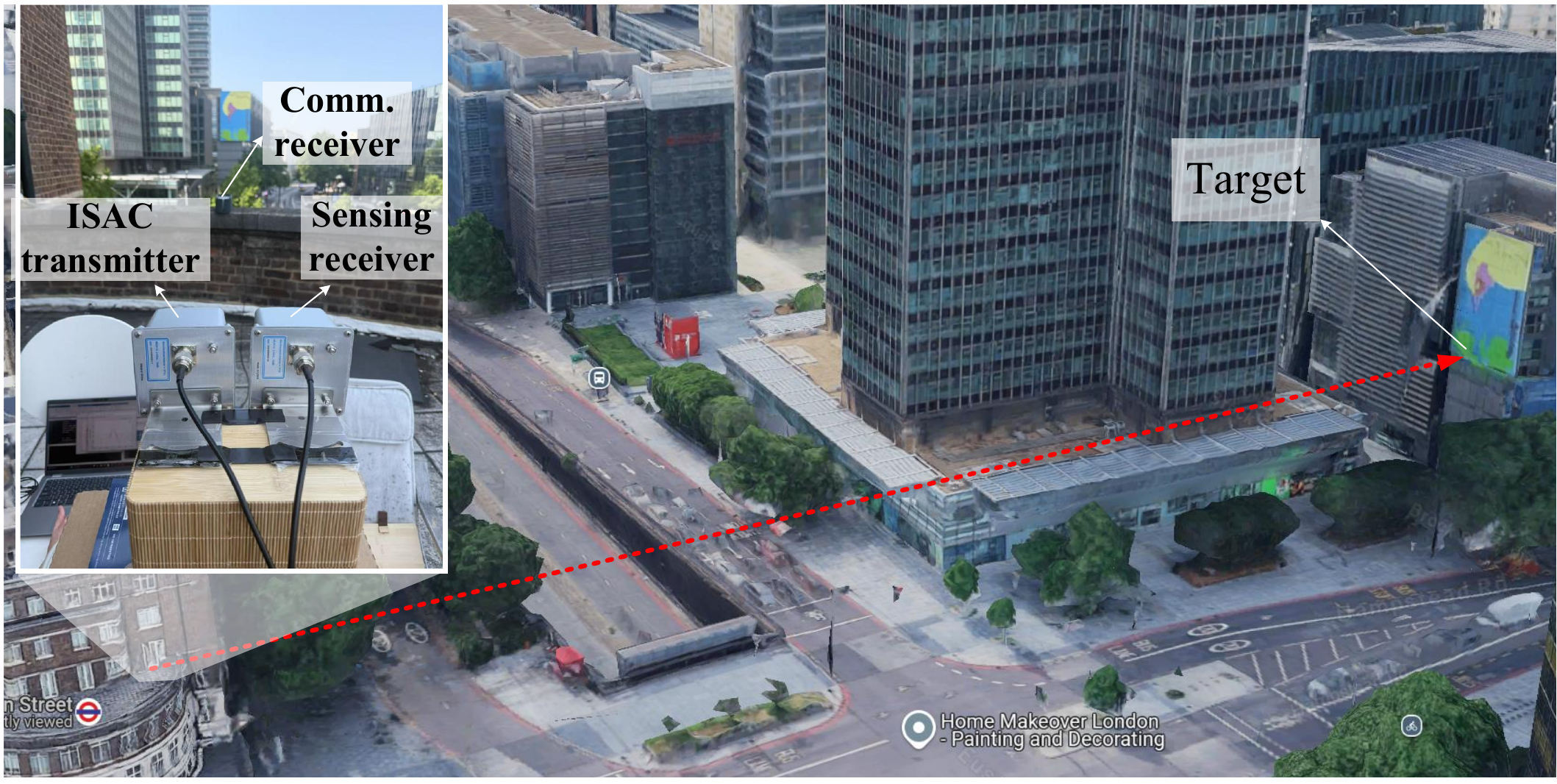}
	\vspace{0mm}
	\caption{Scenario of experiments (google maps).}
	\label{figure9}
\end{figure}
As shown in Fig. \ref{figure9}, the field trial was conducted in a street-canyon environment. An ISAC node built with ADALM-Pluto SDRs was installed on a rooftop, with a co-located receiver forming a monostatic sensing channel, while a second Pluto placed down-range acted as the communication receiver. The dashed line indicates the propagation path toward the building façade that served as the radar target (about 132.6 m).
During the experiments we transmitted OFDM frames with the proposed constellation selection strategy active at the transmitter. For each configuration we recorded sensing outputs including the range profile, the target distance estimate, and aggregate sensing metrics, together with communication outcomes including throughput and bit error rate. This controlled protocol enables a fair comparison of ranging accuracy and communication quality and isolates the effect of constellation selection.

\begin{figure}[t]
	\centering
	\includegraphics[width=8cm]{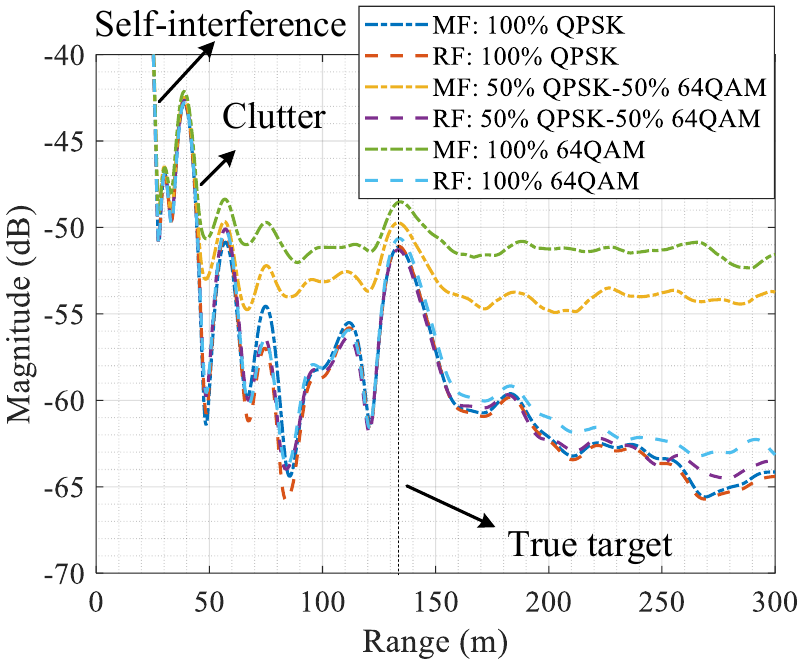}
\vspace{0mm}
	\caption{Measured range profiles for MF and RF under different constellation mixtures.}
	\label{figure10}
\end{figure}
Fig. \ref{figure10} shows range profiles (magnitude in dB) for MF and RF under three constellation mixes, including 100\% QPSK, 50\% QPSK and 50\% 16QAM, as well as 100\% 16QAM. As the share of 16QAM increases, the MF sidelobe level rises markedly across the entire delay axis, producing a pronounced clutter bump (e.g., near 50-80 m) and a monotonic degradation of the peak-to-sidelobe ratio (PSLR) around the true target. This follows from the larger amplitude variance / higher fourth-order moment of non-constant-envelope 16QAM, which injects more off-peak correlation energy that MF passes to the output. RF exhibits the same qualitative trend but with a milder increase, because its reciprocal weighting attenuates part of the off-peak energy while suffering a modest rise in effective noise-power variance. Overall, higher 16QAM fractions elevate sidelobes (MF more than RF) and reduce PSLR, while the target peak location remains unchanged.

\begin{figure}[t]
	\centering
	\includegraphics[width=8cm]{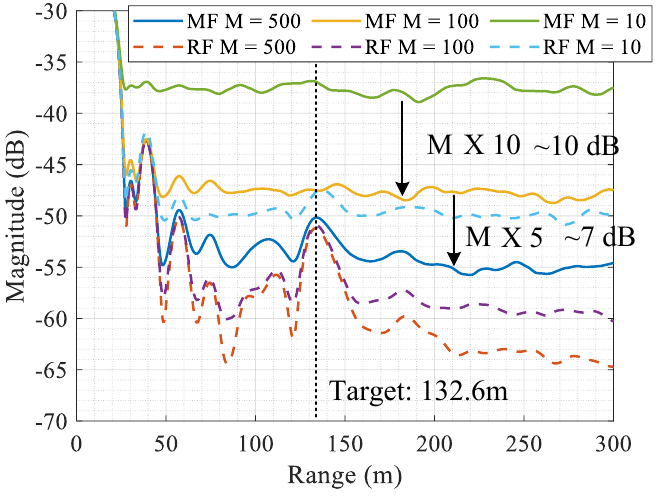}
	\vspace{0mm}
	\caption{Measured coherent-processing gain versus the number of integrated symbols $M$ for a 50\% QPSK + 50\% 16QAM waveform.}
	\label{figure14}
\end{figure}
As shown in Fig. \ref{figure14}, we compare coherent integration over $M\in\{10,100,500\}$ symbols for both the MF and the RF processor. The range peak at 132.6 m is preserved across all cases, while the sidelobe floor decreases with $M$ in accordance with the coherent-processing law $10\log_{10} M$. Specifically, increasing $M$ from 10 to 100 yields an approximate 10 dB reduction of sidelobes, and from 100 to 500 yields an additional 7 dB reduction of sidelobes. This matches the theoretical expression derived for the MF sidelobe level. The trends for RF are consistent: enlarging $M$ effectively mitigates the excess noise introduced by the random communication waveform, leading to noticeably lower sidelobes at larger $M$. {{Overall, the measurements indicate a high-SNR regime where sidelobe suppression scales predictably with $M$, corroborating the analysis. Moreover, as the transmit power or coherent integration gain increases, thermal noise becomes less dominant and the residual limitation of MF is increasingly the clutter-induced sidelobe floor, whereas RF benefits more directly from the improved effective SNR. This is why the RF advantage becomes more visible in the high-power or large-$M$ regime.}}

\begin{figure}[t]
	\centering
	\includegraphics[width=7cm]{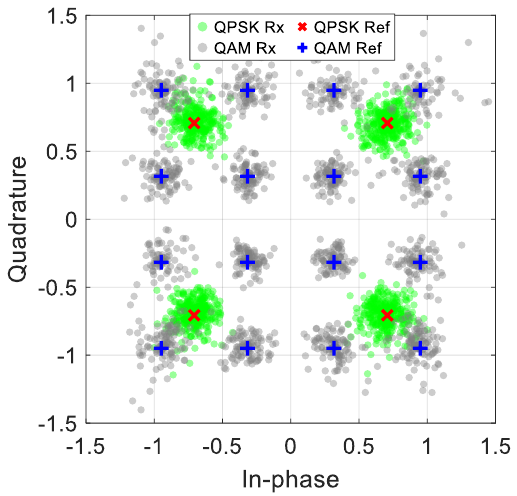}
	\vspace{0mm}
	\caption{Received I/Q scatter plot for a mixed-constellation OFDM frame at the communication receiver. }
	\label{figure13}
\end{figure}
Fig. \ref{figure13} shows the received I/Q scatter (constellation diagram) for a mixed-constellation OFDM frame in which some subcarriers use QPSK and others 16QAM. {{At the communication receiver, pilot symbols embedded in each OFDM frame are used for subcarrier-wise channel estimation. Since the CP is longer than the dominant channel delay spread in the experimental setup, standard frequency-domain one-tap equalization is applied after CP removal and FFT.}} The detected point clouds align tightly with the reference constellation markers, indicating well-separated decision regions and bounded error vector magnitude (EVM) for both formats. Since the receiver knows the per-subcarrier constellation map (e.g., via pre-designed protocol and control signaling), it applies the appropriate slicer on each tone and decodes reliably despite the heterogeneous modulation. This confirms that the proposed mixed-constellation allocation, used to benefit sensing, does not compromise communication demodulation performance.

\begin{figure}[t]
	\centering
	\includegraphics[width=7.4cm]{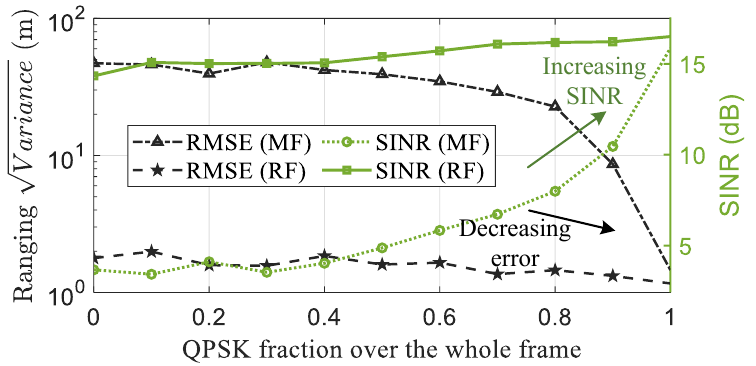}
	\vspace{0mm}
	\caption{Sensing performance versus the QPSK fraction under uniform power allocation, where the remaining subcarriers use 16QAM. }
	\label{figure11}
\end{figure}
{{Fig. \ref{figure11} plots the sensing performance versus the QPSK fraction in a mixed QPSK and 16QAM frame where the remaining symbols are 16QAM. As the QPSK fraction increases, effective SINR rises and ranging error falls for both MF and RF; the improvement is especially marked for MF because 16 QAM's amplitude dispersion elevates sidelobe/clutter power and inflates estimator variance. It can be found that replacing 16QAM with the constant-envelope QPSK suppresses these sidelobes and sharply improves MF accuracy. RF shows a much milder dependence on the mix because its reciprocal weighting already mitigates symbol-induced sidelobes, and thus its improvement with more QPSK comes mainly from a reduction in the effective noise-power term. As the frame approaches all-QPSK, both MF and RF achieve their best accuracy.}}

\begin{figure}[t]
	\centering
	\includegraphics[width=7.4cm]{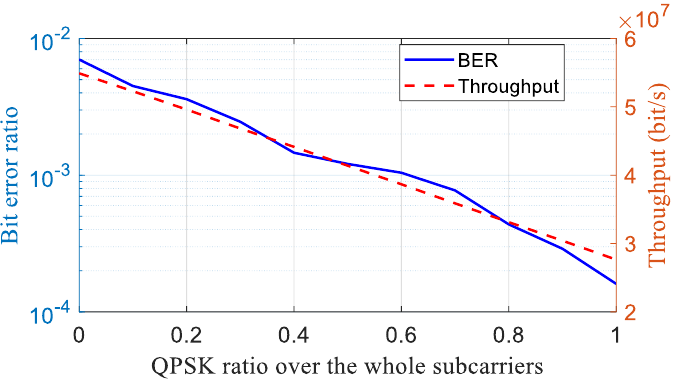}
	\vspace{0mm}
	\caption{Communication BER and throughput versus QPSK fraction (remaining symbols use 16QAM).}
	\label{figure12}
\end{figure}

Fig. \ref{figure12} plots BER (left axis, blue) and aggregate throughput (right axis, red) versus the fraction of QPSK used across all subcarriers in a frame. Throughput is reported as $R_{\mathrm{com}} (1-\mathrm{BER})$. As the QPSK fraction over the whole frame increases, the BER decreases monotonically ($\approx 10^{-2}\!\to\!10^{-4}$), reflecting the higher robustness of the lower-order constellation. However, as the concomitant loss of spectral efficiency dominates, the net throughput declines steadily because the BER reduction is insufficient to offset the rate penalty from replacing higher-order modulation with QPSK. This trend highlights the expected reliability-rate tradeoff under the tested SNR and coding settings. {{As the QPSK fraction changes, Fig. \ref{figure15} reveals a trade-off between sensing SINR output after receive filtering and achievable throughput, evidencing that constellation selection governs the sensing-communication balance.}}

\begin{figure}[t]
	\centering
	\includegraphics[width=6.4cm]{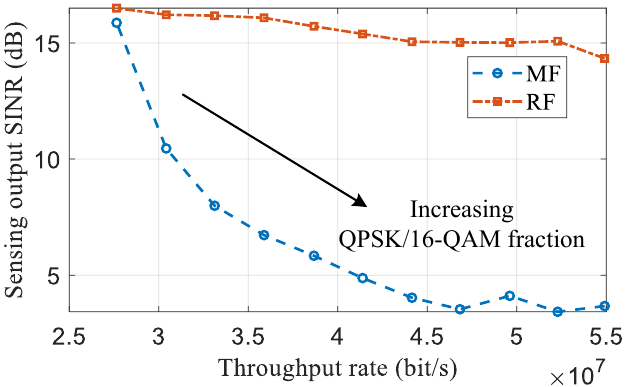}
	\vspace{0mm}
	\caption{Measured trade-off between sensing performance and communication throughput as the constellation mixture varies.}
	\label{figure15}
\end{figure}

\section{Conclusion and Future Work}
This paper developed a standards-compatible, OFDM-based ISAC framework that makes the sensing-communication tradeoff explicit through constellation combination and power shaping. Treating payload symbols as a random process, we established moment-driven sensing laws: MF sidelobes are governed by the per-subcarrier fourth moment, while RF noise amplification is governed by the inverse second moment. The analysis extends to multi-symbol coherent integration, predicts the expected processing gain, and enables super-resolution delay estimation.
Guided by these laws, we proposed a low-complexity co-design that retains a finite, off-the-shelf constellation set and shapes power subject to rate and BER constraints. We proved structural optimality in flat fading channels, any Pareto point activates at most three constellations, and provided convex-equivalent formulations with an exact MISOCP benchmark. For frequency-selective channels, a bilevel algorithm with closed-form inner updates preserves binary constellation decisions and attains near-optimal performance at a fraction of the cost.
Extensive simulations and experiments corroborate the analysis: Mixed-constellation frames demodulate reliably with appropriate subcarrier power control; increasing the constant-envelope share suppresses sidelobes; and performance crosses over from MF at low SNR to RF at high SNR. These results show that constellation-aware power shaping is a practical, standards-compliant lever for ISAC. {{Future work will consider  jointly optimizing subcarrier activation, constellation selection, and power allocation, so as to characterize the resulting trade-off among sensing performance, communication rate, and subcarrier utilization.}}

\section*{Appendix A: \textsc{Proof of Proposition \ref{EISLexpression}}}

According to the definition, $r_k$ can be further expressed as 
\begin{equation} \vspace{0mm}
	\begin{aligned}
		r_k
		&= \sum_{t=0}^{N-1}
		\left(
		\frac{1}{\sqrt{N}}\sum_{n=0}^{N-1}\sqrt{P_n}\, S[n]^{*}\, e^{-j\frac{2\pi}{N} n t}
		\right) \\
		&
		\left(
		\frac{1}{\sqrt{N}}\sum_{n'=0}^{N-1}\sqrt{P_{n'}}\, S[n']\, e^{j\frac{2\pi}{N} n' (t+k)}
		\right) \\
		&= \frac{1}{N}\sum_{n=0}^{N-1}\sum_{n'=0}^{N-1}
		\sqrt{P_n P_{n'}}\, S[n]^{*} S[n']\,
		e^{j\frac{2\pi}{N} n' k}
		\sum_{t=0}^{N-1} e^{j\frac{2\pi}{N}(n'-n)t} \\
		&  = \sum_{n=0}^{N-1} P_n |S[n]|^2 e^{j\frac{2\pi}{N}kn}.
	\end{aligned}
\vspace{0mm} \end{equation}
Then, we have
\begin{equation} \vspace{0mm}
	\label{eq:A.start}
	\begin{aligned}
		|r_k|^2
		&= \Bigl|\sum_{n=0}^{N-1} P_n |S[n]|^2 e^{j\frac{2\pi}{N}kn}\Bigr|^2 \\
		&= \sum_{n=0}^{N-1}\sum_{n'=0}^{N-1}    
				P_n P_{n'} |S[n]|^2 |S[n']|^2 e^{j\frac{2\pi}{N}k(n-n')}.
	\end{aligned}
\vspace{0mm} \end{equation}
Taking expectations and splitting the double sum into diagonal ($n=n'$) and off-diagonal ($n\neq n'$) parts yields
\begin{equation} \vspace{0mm}
	\label{eq:A.Er2}
	\mathbb{E}\!\left[|r_k|^2\right]
	= \sum_{n=0}^{N-1} \mu_{4,n} P_n^2
	\;+\; \sum_{\substack{n,n'=0\\ n\neq n'}}^{N-1}
	P_n P_{n'} e^{j\frac{2\pi}{N}k(n-n')},
\vspace{0mm} \end{equation}
where we used $\mathbb{E}[|S[n]|^4]=\mu_{4,n}$ and, by independence and unit-power normalization,
$\mathbb{E}[|S[n]|^2|S[n']|^2]=1$ for $n\neq n'$.

Summing \eqref{eq:A.Er2} over $k=1,\ldots,N-1$, the first (diagonal) term is $k$-independent and contributes
\begin{equation} \vspace{0mm}
\sum_{k=1}^{N-1}\sum_{n=0}^{N-1} \mu_{4,n} P_n^2
= (N-1)\sum_{n=0}^{N-1} \mu_{4,n} P_n^2.
\vspace{0mm} \end{equation}
For the off-diagonal case $n \neq n'$, the independence of $S[n]$ and $S[n']$ implies $\mathbb{E}[|S[n]|^2 |S[n']|^2] = 1$, resulting in the term
$
\sum_{n \neq n'} P_n P_{n'} \, e^{j\frac{2\pi}{N} k(n-n')}.
$
Combining these results, the expectation of $|r_k|^2$ can be expressed as
\begin{equation} \vspace{0mm}\label{Er2Expression}
	\mathbb{E}[|r_k|^2] = \sum_{n=0}^{N-1} P_n^2 \,\mu_{4,n} 
	+ \sum_{n \neq n'} P_n P_{n'} \, e^{j\frac{2\pi}{N} k(n-n')}.
\vspace{0mm} \end{equation}
Noting that 
$\left|\sum_n P_n e^{j\frac{2\pi}{N}kn}\right|^2
= \sum_n P_n^2 + \sum_{n\neq n'} P_n P_{n'} e^{j\frac{2\pi}{N}k(n-n')}$,
we have
$\sum_{k=1}^{N-1}\sum_{n\neq n'} P_n P_{n'} e^{j\frac{2\pi}{N}k(n-n')}
= \sum_{k=1}^{N-1}\left|\sum_n P_n e^{j\frac{2\pi}{N}kn}\right|^2-(N-1)\sum_n P_n^2$.
Combining with the diagonal part yields
$(N-1)\sum_n(\mu_{4,n}-1)P_n^2+\sum_{k=1}^{N-1}\left|\sum_n P_n e^{j\frac{2\pi}{N}kn}\right|^2$,
which leads to the ESL formula after applying Parseval theorem.

\section*{Appendix B: \textsc{Proof of Theorem \ref{TheoremSINRcoherent}}}

The expected ACF at lag $k$ is
\begin{equation} \vspace{0mm}
	\mathbb{E}[\overline{r}_k] = \sum_{n=0}^{N-1}P_n e^{j\frac{2\pi}{N}n k},
\vspace{0mm} \end{equation}
since $\mathbb{E}[|S^{(m)}[n]|^2]=1$. Therefore, the mean of $\overline{r}_k$ remains unchanged compared to the single-symbol case. That is, the mean is invariant under averaging.

To compute $\mathbb{E}[|\overline{r}_k|^2]$, we analyze the variance of $r_k^{(m)}$:
\begin{align}
	\mathrm{Var}(r_k^{(m)}) &= \mathbb{E}[|r_k^{(m)}|^2] - |\mathbb{E}[r_k^{(m)}]|^2, \\
	\mathbb{E}[|r_k^{(m)}|^2] &= \sum_{n=0}^{N-1}P_n^2\mu_{4,n} + \sum_{n\neq n'}P_nP_{n'} e^{j\frac{2\pi}{N}k(n-n')}, \\
	|\mathbb{E}[r_k^{(m)}]|^2 &= \left|\sum_{n=0}^{N-1}P_n e^{j\frac{2\pi}{N}n k}\right|^2.
\end{align}
Using the identity
\begin{equation} \vspace{0mm}
	\sum_{n\neq n'}P_nP_{n'} e^{j\frac{2\pi}{N}k(n-n')} = \left|\sum_{n=0}^{N-1}P_n e^{j\frac{2\pi}{N}n k}\right|^2 - \sum_{n=0}^{N-1}P_n^2,
\vspace{0mm} \end{equation}
we simplify the variance as
\begin{equation} \vspace{0mm}
	\mathrm{Var}(r_k^{(m)}) = \sum_{n=0}^{N-1}P_n^2(\mu_{4,n}-1).
\vspace{0mm} \end{equation}
Since the autocorrelations $r_k^{(m)}$ across different symbols are independent, the variance of their average is reduced by a factor of $1/M$, i.e.,
\begin{equation} \vspace{0mm}
	\mathrm{Var}(\overline{r}_k) = \frac{1}{M} \mathrm{Var}(r_k^{(m)}).
\vspace{0mm} \end{equation}
This implies that the variance of the autocorrelation follows directly from the single-symbol case, scaled by $1/M$ due to averaging.

Therefore, the second-order moment of the averaged ACF becomes
\begin{equation} \vspace{0mm}
	\mathbb{E}[|\overline{r}_k|^2] = \frac{1}{M}\sum_{n=0}^{N-1}P_n^2(\mu_{4,n}-1) + \left|\sum_{n=0}^{N-1}P_n e^{j\frac{2\pi}{N}n k}\right|^2.
\vspace{0mm} \end{equation}
Using Parseval's identity,
\begin{equation} \vspace{0mm}
	\sum_{k=0}^{N-1}\left|\sum_{n=0}^{N-1}P_n e^{j\frac{2\pi}{N}n k}\right|^2 = N\sum_{n=0}^{N-1}P_n^2,
\vspace{0mm} \end{equation}
we extract the sidelobe power as
\begin{equation} \vspace{0mm}
	\sum_{k=1}^{N-1}\left|\sum_{n=0}^{N-1}P_n e^{j\frac{2\pi}{N}n k}\right|^2 = N\sum_{n=0}^{N-1}P_n^2 - \left(\sum_{n=0}^{N-1}P_n\right)^2.
\vspace{0mm} \end{equation}
Thus, the integrated sidelobe over nonzero lags is
\begin{equation} \vspace{0mm}
	\sum_{k=1}^{N-1}\mathbb{E}[|\overline{r}_k|^2] \!= \! \frac{N-1}{M}\sum_{n=0}^{N-1}P_n^2(\mu_{4,n}-1) 
	+ N\sum_{n=0}^{N-1}P_n^2 - (\sum_{n=0}^{N-1}P_n)^2.
\vspace{0mm} \end{equation}
Substituting these into the SINR definition completes the proof.

\section*{Appendix C: Proof of Theorem \ref{TheoremSINRcoherentRF}}

{{The Tikhonov-regularized version mentioned in the main text is a practical numerical safeguard. For analytical tractability, this appendix derives the SNR for the ideal reciprocal filter $Y[n]/X[n]$. }}

First, for single symbol ($M=1$),
RF per subcarrier is represented by
\begin{equation} \vspace{0mm}
	\tilde Y^{(m)}[n]\;=\;\frac{Y^{(m)}[n]\;X^{(m)}[n]^*}{|X^{(m)}[n]|^2}
	=\sum_{q=1}^Q \alpha_q\,e^{-j\frac{2\pi}{N}n\tau_q}
	+\underbrace{\frac{Z^{(m)}[n]}{X^{(m)}[n]}}_{\triangleq\,\widetilde Z^{(m)}[n]} .
\vspace{0mm} \end{equation}
Conditioned on $X^{(m)}[n]$, $\widetilde Z^{(m)}[n]\sim\mathcal{CN}\!\big(0,\sigma_z^2/|X^{(m)}[n]|^2\big)$, hence
\begin{equation} \vspace{0mm}
	\mathbb E\!\left[\,|\widetilde Z^{(m)}[n]|^2\,\right]
	=\sigma_z^2\,\mathbb E\!\left[|X^{(m)}[n]|^{-2}\right]
	=\sigma_z^2\,\frac{\nu_{-2,n}}{P_n}.
\vspace{0mm} \end{equation}
{{Transforming back to time via the unitary IDFT, we obtain
		\begin{equation} \vspace{0mm}
			\begin{aligned}
				\widetilde y^{(m)}[t]
				&=\frac{1}{\sqrt N}\sum_{n=0}^{N-1}\widetilde Y^{(m)}[n]\,e^{j\frac{2\pi}{N}nt} \\
				&=\sum_{q=1}^Q \alpha_q
				\left(
				\frac{1}{\sqrt N}\sum_{n=0}^{N-1}e^{j\frac{2\pi}{N}n(t-\tau_q)}
				\right)
				+\;w^{(m)}[t] \\
				&=\sum_{q=1}^Q \alpha_q \sqrt N\,\delta_N[t-\tau_q] + w^{(m)}[t],
			\end{aligned}
			\vspace{0mm} \end{equation}
		where $\delta_N[\ell]$ denotes the periodic Kronecker delta, i.e.,
		\begin{equation} \vspace{0mm}
			\delta_N[\ell]=
			\begin{cases}
				1, & \ell=0 \ (\mathrm{mod}\ N),\\
				0, & \text{otherwise}.
			\end{cases}
			\vspace{0mm} \end{equation}
		Also,
		\begin{equation} \vspace{0mm}
			w^{(m)}[t]=\frac{1}{\sqrt N}\sum_{n=0}^{N-1}\widetilde Z^{(m)}[n]\,e^{j\frac{2\pi}{N}nt},
			\vspace{0mm} \end{equation}
		with
		\begin{equation} \vspace{0mm}
			\mathrm{Var}\big(w^{(m)}[t]\big)
			=\frac{1}{N}\sum_{n=0}^{N-1}\sigma_z^2\,\frac{\nu_{-2,n}}{P_n}.
			\vspace{0mm} \end{equation}
		At the target bin $t=\tau_q$, the signal power is therefore
		$|\alpha_q|^2\,|\sqrt N\,\delta_N[0]|^2=|\alpha_q|^2 N$.}}
 Therefore,
\begin{equation} \vspace{0mm}
	\mathrm{SNR}_{q,\;M=1}^{\mathrm{RF}} = \frac{\sigma_{\alpha_q}^2 N^2}{{\sigma_z^2}{}\sum_{n=0}^{N-1}\dfrac{\nu_{-2,n}}{P_n}}.
\vspace{0mm} \end{equation}

Considering coherent averaging over $M$ symbols, the time-domain noise after IDFT,
$\overline w[t]=\frac{1}{\sqrt N}\sum_n \overline Z[n]e^{j\frac{2\pi}{N}nt}$,
has
\begin{equation} \vspace{0mm}
	\mathrm{Var}\big(\overline w[t]\big)
	=\frac{1}{N}\sum_{n=0}^{N-1}\frac{1}{M}\sigma_z^2\,\frac{\nu_{-2,n}}{P_n}
	=\frac{\sigma_z^2}{MN}\sum_{n=0}^{N-1}\frac{\nu_{-2,n}}{P_n}.
\vspace{0mm} \end{equation}
The peak signal power remains $\sigma_{\alpha_q}^2 N$, hence
\begin{equation} \vspace{0mm}
		\mathrm{SNR}_q^{\mathrm{RF}}
		=\frac{\sigma_{\alpha_q}^2  N^2 M}
		{{\sigma_z^2}\displaystyle\sum_{n=0}^{N-1}\frac{\nu_{-2,n}}{P_n}}.
\vspace{0mm} \end{equation}
This thus completes the proof.

\section*{Appendix D: \textsc{Proof of Theorem \ref{TwoconstellationIsEnough}}}
\vspace{0mm}

Let $(\boldsymbol{\eta}^\star,\mathbf P^\star)$ be an optimal solution of
\eqref{eq:constraints}. Fixing $\mathbf P^\star$, we optimize only over
$\boldsymbol{\eta}$. Since maximizing $\Gamma_s$ is equivalent to minimizing
an affine function of $\boldsymbol{\eta}$ for both RF and MF, the resulting
subproblem is a linear program.  

Introduce a slack variable $\xi \ge 0$ and rewrite the rate constraint as
\begin{equation}
	\sum_{j=1}^J \eta_j R_j - \xi  = R_{\min}.
\end{equation}
Then the $\boldsymbol{\eta}$-subproblem can be written in standard form as
\begin{equation}
	\min_{\mathbf y}\ \mathbf c^\top \mathbf y
	\qquad
	\text{s.t.}\qquad
	A\mathbf y=\mathbf b,\quad \mathbf y\ge 0,
\end{equation}
where
$
\mathbf y=[\eta_1,\ldots,\eta_J,\xi ]^\top,
$
$
\mathbf c=[c_1,\ldots,c_J,0]^\top$,
$\mathbf b= [1, P_{\rm ave}, R_{\min}]^\top.
$
Here, $c_j=\nu_{-2,j}/P_j^\star$ for RF, and
$
c_j=\left(\frac{N}{M}(\mu_{4,j}-1)+\frac{N^2}{N-1}\right)(P_j^\star)^2
$
for MF.
and
\begin{equation}
	A=
	\begin{bmatrix}
		1 & 1 & \cdots & 1 & 0\\
		P_1^\star & P_2^\star & \cdots & P_J^\star & 0\\
		R_1 & R_2 & \cdots & R_J & -1
	\end{bmatrix}.
\end{equation}
Hence, $\operatorname{rank}(A)\le 3$. By the fundamental theorem of linear
programming, there exists an optimal basic feasible solution. Therefore, at
most three components of $[\eta_1,\ldots,\eta_J,\xi]$ are nonzero.
If the rate constraint is active at optimality, then $\xi =0$, and thus at most
three entries of $\boldsymbol{\eta}$ are nonzero. If the rate constraint is
inactive at optimality, then $\xi >0$, so at most two entries of
$\boldsymbol{\eta}$ can be nonzero.It is readily seen that the above argument holds for any fixed $\mathbf P$.
Therefore, the same conclusion also holds for the joint optimization of
$\boldsymbol{\eta}$ and $\mathbf P$.

Hence, there exists an optimal solution of the original problem involving at
most two active constellation types when the rate constraint is inactive, and
at most three active constellation types when the rate constraint is active. 

\vspace{0mm}
\footnotesize  	
\bibliography{mybibfile}
\bibliographystyle{IEEEtran}

\end{document}